\documentclass[runningheads,envcountsame]{llncs}
\usepackage[T1]{fontenc}

\usepackage{amsmath}
\usepackage{amssymb}
\usepackage{mathtools}
\usepackage{thmtools, thm-restate}
\usepackage{bussproofs}
\usepackage{enumerate}
\usepackage{cleveref}
\usepackage{bbding}

\let\llncsproof\proof
\renewcommand{\proof}[1][]{%
  \ifx!#1!\else\renewcommand{\proofname}{#1}\fi
  \llncsproof
}

\newcommand{\abs}[1]{\left|#1\right| }
\newcommand{\suchthat}{\ |\ }

\newcommand{\mset}[1]{\left\{#1\right\}}
\newcommand{\union}{\cup}

\newcommand{\seq}[1]{\overline{#1}}

\renewcommand{\iff}{\equiv}
\newcommand{\imp}{\Rrightarrow}
\newcommand{\limp}{\rightarrow}
\newcommand{\liff}{\leftrightarrow}
\newcommand{\Lor}{\bigvee}
\newcommand{\Land}{\bigwedge}
\newcommand{\oeq}{\simeq}
\newcommand{\noeq}{\not\oeq}

\newcommand{\resclosure}[2]{\mathrm{Res}_{#1}({#2})}
\newcommand{\resclosurepremise}[2]{\mathrm{ResU}_{#1}(#2)}
\newcommand{\resclosurepredicate}[1]{\mathrm{pResU}_{#1}}

\newcommand{\definedas}{:=}

\newcommand{\N}{\mathbb{N}}

\newcommand{\Exists}[1]{\exists #1 \, }
\newcommand{\Forall}[1]{\forall #1 \, }

\newcommand{\witness}{\alpha}

\usepackage[disable]{todonotes}
\newboolean{todoactive}
\setboolean{todoactive}{false}
\let\todonote=\todo
\newcommand{\todocolor}{red}
\renewcommand{\todo}[1]{%
  \ifthenelse{\boolean{todoactive}}%
  {{\color{\todocolor}\footnote{\todonote[inline,size=\footnotesize,textcolor=\todocolor,bordercolor=white,backgroundcolor=white,noinlinepar]{TODO #1}}}}%
  {}%
}

\begin{document}

\title{Computing Witnesses Using the SCAN Algorithm \\ (Extended Preprint)}

\author{Fabian Achammer\inst{1}\orcidID{0009-0002-3799-6393}
\and Stefan Hetzl \inst{1}\orcidID{0000-0002-6461-5982}
\and Renate A.~Schmidt \inst{2}\orcidID{0000-0002-6673-3333}}
\authorrunning{F. Achammer \and S. Hetzl \and R. A. Schmidt}

\institute{Institute for Discrete Mathematics and Geometry, TU Wien, Vienna, Austria
\email{\{fabian.achammer,stefan.hetzl\}@tuwien.ac.at}\\
\and
Department of Computer Science, University of Manchester, Manchester, UK\\
\email{renate.schmidt@manchester.ac.uk}}

\maketitle

\begin{abstract}
  Second-order quantifier-elimination is the problem of finding, given a formula with second-order quantifiers, a logically equivalent first-order formula.
  While such formulas are not computable in general, there are practical algorithms and subclasses with applications throughout computational logic.
  One of the most prominent algorithms for second-order quantifier elimination is the SCAN algorithm which is based on saturation theorem proving.
  In this paper we show how the SCAN algorithm on clause sets can be extended to solve a more general problem: namely, finding an instance of the second-order quantifiers that results in a logically equivalent first-order formula.
  In addition we provide a prototype implementation of the proposed method.
  This work paves the way for applying the SCAN algorithm to new problems in application domains such as modal correspondence theory, knowledge representation, and verification.
\end{abstract}

\section{Introduction}
\emph{Second-order quantifier elimination (SOQE)}~\cite{Gabbay92Quantifier,Gabbay08Second} is the problem
of, given a formula with second-order quantifiers 
$\Exists{\seq{X}} \varphi$ where $\varphi$ is first-order, find a first-order formula
$\psi$ such that
$$
\Exists{\seq{X}} \varphi \iff \psi.
$$
For example, one can show that $\varPhi = \Exists{X} (X(a) \land \Forall{u} (X(u) \limp B(u)))$
is equivalent to $B(a)$.
An algorithm for eliminating existential quantifiers can
be used to eliminate universal quantifiers by writing
$\Forall{\seq{X}}$ as $\neg \exists \seq{X} \neg$.

SOQE is an important topic in logic, automated reasoning and artificial
intelligence that has real-world applications in diverse areas.
It has been used to automate correspondence theory for modal
logic~\cite{Gabbay92Quantifier}, which has led
to new attention and results in correspondence theory of various
algebras and logics~\cite{BrinkGabbayOhlbach95,Goranko03Scan,ConradieGorankoVakarelov06a,ConradieGhilardiPalmigiano14}.
SOQE can be used to transform a knowledge base (or set of formulas, given by~$\varphi$) into
a restricted view $\psi$ of the knowledge base in which all occurrences
of the predicate symbols $\overline X$ have been eliminated.
This reduction has also been referred to as
forgetting~\cite{LinReiter94} or projection and studied in the form
of uniform interpolation or strongest necessary condition in knowledge
representation~\cite{KonevWaltherWolter09,KoopmannSchmidt13a,KoopmannSchmidt13c,Delgrande17,FereeVanDerGiessenEtAl24,DohertyLukaszewiczSzalas01} 
 and answer set programming~\cite{EiterKernIsberner19,GoncalvesKnorrLeite23}.
Forgetting offers solutions to several important applications and manipulations of
ontologies: ontology extraction, ontology creation, reuse and comparison, 
and abductive reasoning~\cite{ChenAlghamdiEtAl19a,LudwigKonev14,LiuLuEtAl21,KoopmannDelPintoEtAl20}.
In automated reasoning an application of SOQE is (predicate) symbol
elimination~\cite{Gabbay92Quantifier,HoderKovacsVoronkov10,KhasidashviliKorovin16,PeuterSofronieStokkermans21}.
Another very promising application domain is agent communication,
requiring knowledge sharing among agents~\cite{DohertySzalas04,ToluhiSchmidtParsia21}.

Two prominent approaches to computing SOQE-solutions 
are the SCAN algorithm introduced
in~\cite{Gabbay92Quantifier} and the DLS algorithm
introduced in~\cite{Doherty97Computing} and extended
in~\cite{Doherty98General,Nonnengart98Fixpoint}.
SCAN is based on the idea of computing all logical
consequences of $\Exists{\seq{X}} \varphi$ and omit those formulas that contain
predicate variables from $\seq{X}$.
It transforms the first-order part $\varphi$ into clausal
normal form and computes the closure under a constraint resolution calculus,
performing inferences only on $\seq{X}$-literals. During this process 
clauses with $\seq{X}$-literals for which sufficiently many inferences
have been performed are deleted~\cite{Engel96Quantifier,Ohlbach96SCAN}.
If the algorithm terminates, it returns a set of first-order consequences
which is logically equivalent $\Exists{\seq{X}} \varphi$.
Reverse Skolemization might be applied,
if during the clause form transformation Skolemization was performed.
SCAN has been shown complete for computing frame correspondence properties
for the class of modal Sahlqvist axioms~\cite{Goranko03Scan}.
The SCAN algorithm is closely related to hierarchic
superposition~\cite{Bachmair94Refutational,BaumgartnerWaldmann19}.

The other prominent approach, the DLS algorithm, is based on a result of
Ackermann~\cite{Ackermann35Untersuchungen}.
Known as Ackermann's Lemma, it
states a condition under which a predicate variable is eliminable from a second-order formula $\exists X \varphi$.
The idea of the DLS algorithm is to use equivalence preserving
operations to bring the problem into a form
where Ackermann's Lemma is applicable.

In this paper we are interested in a more general problem which we call
\emph{witnessed second-order quantifier elimination (WSOQE)}: Given a formula $\Exists{\seq{X}} \varphi$ with first-order~$\varphi$, find first-order predicates
(formulas) $\seq{\witness}$ satisfying this \emph{WSOQE-condition}
\begin{equation*}
  \label{wsoqe-condition}
  \Exists{\seq{X}} \varphi \iff \varphi[\seq{X} \leftarrow \seq{\witness}]. \tag{$\ast$}
\end{equation*}
That is, the goal is to find a first-order instantiation $\seq{\witness}$ (in the same language as $\varphi$) of variables $\seq X$
such that an equivalent first-order formula is obtained.
We call $\seq{\witness}$ a \emph{WSOQE-witness for $\Exists{\seq{X}} \varphi$}.
Recall that the second-order formula $\varPhi$ from the beginning of the introduction is equivalent to $B(a)$.
One can check that $\alpha = \lambda u. u \oeq a$ and $\beta = \lambda u. B(u)$ are WSOQE-witnesses for $\varPhi$.

The contribution of this paper is an algorithm to solve WSOQE and compute such witnesses.
Our algorithm, called \emph{WSCAN}, is an extension of the SCAN algorithm with
a post-processing step which extracts a witness from the derivation
of a terminating SCAN run.
The witness construction proceeds in a bottom-up way, computing
iteratively from a
witness at derivation step~$i$ a witness for step~$i-1$.
A prototype has been implemented in version 2.18.1 of the GAPT software package~\cite{GAPT}.
In addition to solving the more general WSOQE problem, the concepts
introduced in this paper provide a new correctness proof for the
SCAN algorithm.

The WSOQE problem bridges the gap between SOQE and another important problem, 
namely, \emph{solving formula equations (FEQ)}: Given an input formula $\Exists{\seq{X}} \varphi$, where~$\varphi$ is first-order find a tuple $\seq{\witness}$ of first-order predicates such that 
$$\models \varphi[\seq{X} \leftarrow \seq{\witness}].$$
Since the earlier example $\varPhi$ is not valid it does not have an FEQ-solution.
However, one can check that the witnesses $\alpha$ and $\beta$ from above are FEQ-solutions for the modified formula $\varPhi' = \Exists{X}(B(a) \limp (X(a) \land \Forall{u} (X(u) \limp B(u))))$.

Going back to studies in the 19th century by~\cite{Schroeder90Vorlesungen}, solving formula equations is one of the oldest and most central problems of logic.
The book~\cite{Rudeanu74Boolean} comprehensively considers the problem in the setting of Boolean algebra.
Solving Boolean equations is closely related to Boolean unification, a subject of thorough study in
computer science, see, e.g.,~\cite{Martin89Boolean} for a survey.
The generalization of this problem from propositional to first-order logic has been made explicit
as early as~\cite{Behmann50Aufloesungsproblem,Behmann51Aufloesungsproblem}.

Today, FEQ is the central problem underlying several areas of computational logic, even though
this is often not made explicit.
For example, solving constrained Horn clauses~\cite{Bjorner15Horn}, a popular formalism in verification,
is a restriction of FEQ, with the existentially quantified variables representing
unknown loop invariants in applications of software verification.
More generally, the problem of inductive theorem proving, e.g., in Peano arithmetic, can be
formulated as a restriction of FEQ, where the predicate variables stand for the unknown
induction formulas.
One can naively solve FEQ for $\Exists{\seq{X}} \varphi$ by enumerating all first-order predicates $\seq{\witness}$ and checking whether $\varphi[\seq{X} \leftarrow \seq{\witness}]$ is valid by using, e.g., a first-order theorem prover.

An algorithm for WSOQE, like the one we introduce in this paper, can be used to solve both SOQE and FEQ:
A WSOQE-witness $\seq{\witness}$ satisfies ($*$) and therefore
$\varphi[\seq{X} \leftarrow \seq{\witness}]$ solves SOQE for $\Exists{\seq{X}} \varphi$.
However, note that WSOQE is not complete for SOQE as there are formulas that have a SOQE-solution, but no WSOQE-witness (see \Cref{sec.discussion} for such an example).
Furthermore, if one has a WSOQE-witness $\seq{\witness}$ for $\Exists{\seq{X}} \varphi$ one can determine whether it is an FEQ-solution for $\Exists{\seq{X}} \varphi$ by checking whether $\varphi[\seq{X} \leftarrow \seq{\witness}]$ is valid which can be done with a first-order theorem prover.
The use of solutions to WSOQE for solving SOQE or FEQ is widespread in the literature:
it is the basis of Ackermann's Lemma~\cite{Ackermann35Untersuchungen}, which provides the foundation for the DLS
algorithm~\cite{Doherty97Computing} for SOQE.
It is central for the decidability of quantifier-free Boolean unification with predicates~\cite{Eberhard17Boolean}
and for the fixed-point theorem for Horn formula equations of~\cite{Hetzl21Fixed,Hetzl21Abstract}.
Solutions to WSOQE are called ELIM-witnesses in~\cite{Wernhard17Boolean}.

The paper is organized as follows.
In \Cref{sec.preliminaries} we introduce some necessary notation.
\Cref{sec.extendedSCAN} describes the SCAN algorithm.
We present our method of extracting witnesses from SCAN derivations in \Cref{sec.witness-extraction}.
\Cref{sec.first-order-witnesses} provides conditions under which these SCAN derivations guarantee first-order witnesses.
The prototype implementation of our method is outlined in \Cref{sec.implementation}.
\Cref{sec.discussion} discusses limitations of our method and outlines avenues for future work.
Full proofs of all results and further examples are given in the appendix.

\section{Preliminaries}\label{sec.preliminaries}

We assume familiarity with classical second-order and first-order logic
with equality, the usual Tarskian semantics and first-order clause logic, and use standard
notation for the connectives, semantic entailment and logical
equivalence.
Two formulas (of first- or second-order logic) are called \emph{(logically) equivalent}, if they have the same Tarski models.
We assume a language with equality $\oeq$ and countably many
individual variables $u,v,w,\dots$, predicate variables
$X, Y, Z, \dots$, individual constant symbols $a, b,
c, \dots$, function symbols $f, g, h, \dots$ and predicate symbols
$A, B, \dots$.
We denote terms by $t, s, r \dots$, formulas by $\varphi, \psi, \dots$, 
literals by $L,L', \dots$, clauses by $C, C', \dots$ 
and clause sets by $N, N', \dots$.
A literal $L$ is called an $\seq{X}$-literal, if $L = Y(\seq{t})$ or $L = \neg Y(\seq{t})$ for some predicate variable $Y$ in $\seq{X}$ and terms $\seq{t}$.
The dualization of a literal $L$ is denoted by $L^{\perp}$ and defined as $\neg X(\seq{t})$ if $L = X(\seq{t})$ and $X(\seq{t})$, if $L = \neg X(\seq{t})$.

When a clause set is used in the context of a formula we mean the
first-order universal closure of its corresponding conjunctive
normal form.
We will sometimes use infinite conjunctions and disjunctions.

The set of \emph{basic expressions} is the union of the set of terms, formulas, clauses and clause set.
The set of \emph{expressions} is the smallest set containing the basic expressions that is closed under
$\lambda$-abstraction, i.e., if $E$ is an expression and $u$ is an individual variable, then $\lambda u. E$ is an expression.

Given expressions $E_1, \dots, E_n$ we write $\seq{E}$ for the tuple $(E_1, \dots, E_n)$.
If $\seq{u} = (u_1, \dots, u_n)$ is a tuple of first-order variables and $E$ an expression, we write $\lambda \seq{u}. E$ for the expression $\lambda u_1. \dots \lambda u_n. E$.
For tuples of terms $\seq{t} = (t_1, \dots, t_n)$ and $\seq{s} = (s_1, \dots, s_n)$ we write $\seq{t} \oeq \seq{s}$  for $t_1 \oeq s_1 \land \dots \land t_n \oeq s_n$ and $\seq{t} \noeq \seq{s}$ for $t_1 \noeq s_1 \lor \dots \lor t_n \noeq s_n$.
If we consider an expression $E(\seq{u})$, 
where $\seq{u}$ is a tuple of first-order variables (or constants),
and $\seq{t}$ is a tuple of terms with the same length as $\seq{u}$, 
then $E(\seq{t})$ denotes the simultaneous substitution of the free variables (or constants) $\seq{u}$ in $E$ by $\seq{t}$.
Additionally, using $E(\seq{u})$ implies that all free variables of~$E$ are among those in $\seq{u}$.
A \emph{predicate expression} $\alpha$ is an expression of the form~$\lambda \seq{u}. \varphi$ where $\varphi$ is a formula.
If $\varphi$ is first-order, we call $\alpha$ a \emph{first-order predicate}.
Predicate expressions $\alpha = \lambda \seq{u}. \varphi(\seq{u})$ and $\beta = \lambda \seq{u}. \psi(\seq{u})$ are \emph{equivalent (in symbols $\alpha \iff \beta$)} if $\models \Forall{\seq{u}} (\varphi(\seq{u}) \liff \psi(\seq{u}))$.

Let $E$ be an expression.
For $1 \leq i \leq d$, let $X_i$ be a predicate variable and $F_i = \lambda \seq{u_i}. \varphi_i(\seq{u_i})$ be a predicate expression of the same arity.
Then we denote by $E[X_1 \leftarrow F_1, \dots, X_d \leftarrow F_d]$ (abbreviated as $E[\seq{X} \leftarrow \seq{F}]$) the expression resulting from $E$ by simultaneously replacing all occurrences $X_i(\seq{t})$ in $E$ by $\varphi_i(\seq{t})$.

For terms $\seq{t} = (t_1, \dots, t_n)$ and $\seq{s} = (s_1, \dots, s_n)$ a substitution $\sigma$ is called a \emph{unifier} of $\seq{t}$ and $\seq{s}$, if $t_i\sigma = s_i\sigma$ for all $i \in \mset{1, \dots, n}$.
A unifier $\sigma$ of $\seq{t}$ and $\seq{s}$ is called a \emph{most general unifier} of $\seq{t}$ and $\seq{s}$, if for all unifiers $\tau$ of $\seq{t}$ and $\seq{s}$ there exists a substitution $\rho$ such that $\tau = \sigma\rho$.

\section{The SCAN Algorithm}\label{sec.extendedSCAN}
The SCAN algorithm takes as input a formula $\Exists{\seq{X}} \varphi$, where $\varphi$ is first-order
and then transforms $\varphi$ into its clausal normal form $N$, which might include a Skolemization step.
Next, $N$ is saturated by a \emph{purification} process:
SCAN picks a clause $C$ and a designated $\seq{X}$-literal $L \in C$.
All non-redundant constraint resolution inferences (defined below) between $C$ (on that literal $L$) and the rest of the clauses are performed.
In addition constraint factoring inferences (defined below) on $C$, constraint elimination and redundancy criteria are used to simplify the resulting clauses.
If all non-redundant inferences have been generated, $C$ is deleted from the clause set and the process is repeated with a new clause and designated $\seq{X}$-literal.
This process might not terminate, but if we reach a point where no clause contains an $\seq{X}$-literal the result is a clause set $N^\ast$ such that $\Exists{\seq{X}} N \iff N^\ast$.
SCAN attempts to reverse the Skolemization step from the beginning, if performed, and returns a first-order formula logically equivalent to $\Exists{\seq{X}} \varphi$.
As we present our results for clause sets, we do not deal with reverse Skolemization.

Note that during the purification process, SCAN not only chooses a clause, but also a literal within this clause.
To give a name to this concept we introduce the notion of \emph{pointed clause} (in analogy to \emph{pointed sets} in mathematics):
\begin{definition}
  Let $L$ be a literal and $C$ a clause.
  A \emph{pointed clause} $P$ is a clause $\underline{L} \lor C$, in which the literal $L$ in $L \lor C$ is underlined.
  We call $L$ the \emph{designated literal of $P$}.
  If $L$ is an $X$-literal, then we say $P$ is an \emph{$X$-pointed} clause.
\end{definition}

We fix an existential prefix of predicate variables $\seq{X}$. The inference rules on clauses used in the SCAN algorithm are:
\begin{flalign*}
  &\textbf{Constraint resolution:}
  &&\AxiomC{$C \lor L(\seq{t})$}
  \AxiomC{$C' \lor L(\seq{s})^\perp$}
  \RightLabel{$\mathrm{Res}$}
  \BinaryInfC{$\seq{t} \noeq \seq{s} \lor C \lor C'$}
  \DisplayProof&\\
  \intertext{where $L$ is an $\seq{X}$-literal and w.l.o.g. $C \lor L(\seq{t})$ and $C' \lor L(\seq{s})^\perp$ are assumed to be variable-disjoint.
  The clause $\seq{t} \noeq \seq{s} \lor C \lor C'$ is called the \emph{resolvent} of this inference.}
  &\textbf{Constraint factoring:}
  &&\AxiomC{$C \lor L(\seq{t}) \lor L(\seq{s})$}
  \RightLabel{$\mathrm{Fac}$}
  \UnaryInfC{$\seq{t} \noeq \seq{s} \lor C \lor L(\seq{t})$}
  \DisplayProof&\\
  \intertext{where $L$ is an $\seq{X}$.
  The clause $\seq{t} \noeq \seq{s} \lor C \lor L(\seq{t})$ is called the \emph{factor} of this inference.}
  &\textbf{Constraint elimination:}
  &&\AxiomC{$\seq{t} \noeq \seq{s} \lor C$}
  \RightLabel{$\mathrm{ConstrElim}$}
  \UnaryInfC{$C\sigma$}
  \DisplayProof&
\end{flalign*}
where $\sigma$ is a most general unifier of $\seq{t}$ and $\seq{s}$.

Our method does not depend on the specific choice of redundancy criteria, as long as the notion of redundancy is
(i) \emph{compatible with entailment}: If a clause $C$ is redundant in a clause set $N$, then $\models N \limp C$ and (ii) \emph{equivalence-preserving}: If $N/N'$ is a redundancy elimination step, then $N \iff N'$.
Common redundancy criteria such as tautology deletion, subsumption deletion, constraint elimination and condensing (replacing a clause by a minimal subclause that subsumes it) all satisfy these properties.

We now give a calculus for the saturation process of SCAN:
\begin{definition}
  We denote by $\mathcal{C}$ the \emph{constraint resolution calculus} operating on clause sets with the following derivation steps:
  \begin{flalign*}
    &\text{\smallskip\noindent\emph{\textbf{Inference:}}}
    && \AxiomC{$N$}
    \RightLabel{$I$}
    \UnaryInfC{$N \union \mset{C}$}
    \DisplayProof&\\
    \intertext{for $I \in \mset{\mathrm{Res},\mathrm{Fac},\mathrm{ConstrElim}}$ if \ $\AxiomC{$\ C_1\cdots\ C_n$}\RightLabel{$I$}\UnaryInfC{$C$}\DisplayProof$
    is one of the inferences from above and $C_1,\dots,C_n\in~N$.}
    &\text{\medskip\noindent\emph{\textbf{Redundancy elimination:}}}
    &&
    \AxiomC{$N$}
    \RightLabel{$\mathrm{RedElim}$}
    \UnaryInfC{$N'$}
    \DisplayProof&\\
    \intertext{if $N' = N \setminus \mset{C}$ with $C \in N$ and $C$ is redundant in $N$.}
    &\text{\smallskip\noindent\emph{\textbf{Extended purity deletion:}}}
    &&
      \AxiomC{$N$}
      \RightLabel{$\mathrm{ExtPurDel}_X^{p}$}
      \UnaryInfC{$N \setminus \mset{C \in N \suchthat \text{$C$ contains $X$}}$}
      \DisplayProof
    &\\
    \intertext{for $p\in \mset{+, -}$, where $X$ is a predicate variable if every clause in $N$ containing~$X$ contains $X$ with polarity $p$.}
    &\text{\medskip\noindent\emph{\textbf{Purified clause deletion:}}}
    &&
    \AxiomC{$N \uplus \mset{P}$}
    \RightLabel{$\mathrm{PurDel}_P$}
    \UnaryInfC{$N$}
    \DisplayProof
    ,&
    \intertext{if $P$ is an $X$-pointed clause, where $X$ is a predicate variable and every resolvent of $P$ with a clause from $N$ is redundant in $N$.
    We call $P$ \emph{the pointed clause} of this purified clause deletion step and also say that \emph{$P$ is purified in $N$}.
    By \emph{purification of $P$} we refer to the process of adding resolvents between $P$ and $N$ until $P$ is purified in $N$ followed by the application of $\mathrm{PurDel}_P$.}
  \end{flalign*}
\end{definition}

One can show that $\mathcal{C}$ preserves equivalence with respect to existential quantification which shows that SCAN is correct:
\begin{restatable}{proposition}{calculusIsSoundAndExistentialEquivalencePreserving}
  \label{correctness-of-scan}
  If $N/N'$ is a $\mathcal{C}$-derivation step, then $\Exists{\seq{X}} N \iff \Exists{\seq{X}} N'$.
\end{restatable}

We record the derivation steps performed during the saturation process using the following definition:
\begin{definition}
  Let $D := (D_i)_{1 \leq i \leq m}$ be a finite sequence of $\mathcal{C}$-derivation steps and let $N,N'$ be clause sets.
  We inductively define the concept of a \emph{$\mathcal{C}$-derivation from~$N$ with conclusion~$N'$}
  by: 
  (i) the empty sequence is a $\mathcal{C}$-derivation from $N$ with conclusion $N$ and 
  (ii) if $D=(D_i)_{1 \leq i \leq m}$ is a $\mathcal{C}$-derivation from~$N$ with conclusion $N'$ 
  and $D_{m+1}$ is a $\mathcal{C}$-derivation step from $N'$ to $N''$, 
  then $(D_i)_{1 \leq i \leq m+1}$ is a $\mathcal{C}$-derivation from $N$ with conclusion $N''.$

  We say \emph{$D$ is a $\mathcal{C}$-derivation from $N$}, if $D$ is a $\mathcal{C}$-derivation from $N$ with conclusion $N'$ for some clause set $N'$.
  We assume the $D_i$ contain enough information so that $N'$ is uniquely determined by $D$ and $N$ so we call $N'$ \emph{the conclusion of $D$ from $N$}.
  For a $\mathcal{C}$-derivation $D$ from $N$ we define the sequence $N(D) := (N(D)_i)_{0 \leq i \leq m}$ of intermediate clause sets in which
  $N(D)_i$ is defined to be the conclusion of the $\mathcal{C}$-derivation $(D_j)_{1 \leq j \leq i}$ from $N$.
  For a tuple of predicate variables $\seq{X}$ we say \emph{$D$ is $\seq{X}$-eliminating from $N$} if the conclusion of $D$ from $N$ contains no predicate variables from $\seq{X}$.
  A pointed clause $P$ is \emph{purified in $D$} if there is a derivation step $D_i$ with $D_i = \mathrm{PurDel}_P$.
\end{definition}

\begin{example}
  \label{ex.main-example.derivation}
  Consider the clause set $N$ with the clauses
  \begin{align*}
    (1)\ B(a,v) \qquad (2)\ X(a) \qquad (3)\ B(u,v) \lor \neg X(u) \lor X(v) \qquad  (4)\ \neg X(c)
  \end{align*}
  and suppose we want to eliminate $X$.
  We denote clauses by their numbers, e.g.,~$1$ refers to clause $B(a,v)$.
  Pointed clauses are referred to by dot notation, e.g., $3.2$ refers to the pointed clause $B(u,v) \lor \underline{\neg X(u)} \lor X(v)$.
  Applying SCAN, some resolvents from $N$ are
  \begin{align*}
    (5)\ a \noeq u \lor B(u,v) \lor X(v) \quad \text{($2$ with $3$)} \quad \text{and} \quad
    (6)\ & a \noeq c & \text{($2$ with $4$)}
  \end{align*}
  Note that clause~$5$ is semantically entailed by clause~$1$ because $1$ subsumes $5$ after constraint elimination.
  Let $D$ be this derivation:
  $$
    \Axiom$\mset{1,2,3,4}\fCenter$
    \RightLabel{$\mathrm{Res}_{2.1,4.1}$}
    \UnaryInf$\mset{1,2,3,4,6}\fCenter$
    \RightLabel{$\mathrm{PurDel}_{2.1}$}
    \UnaryInf$\mset{1,3,4,6}\fCenter$
    \RightLabel{$\mathrm{ExtPurDel}_X^{-}$}
    \UnaryInf$\mset{1,6}\fCenter$
    \DisplayProof
  $$
  The subscripts in $\mathrm{Res}$ indicate the corresponding resolution premises.
  Note that the pointed clause~$2.1$ is purified in $\mset{1,3,4,6}$ since $5$, the resolvent of $2$ with $4$, is redundant in $\mset{1,3,4,6}$.
  Thus the $\mathrm{PurDel}_{2.1}$ step is applicable to $\mset{1,2,3,4,6}$.
  As $\mset{1,6}$ does not contain $X$, we have that $D$ is an $X$-eliminating derivation.
\end{example}

\section{Witness Extraction from SCAN Derivations}
\label{sec.witness-extraction}

In this section we introduce the WSCAN algorithm (`W' for witness), as an extension of SCAN on clause sets by providing a witness for the second-order quantifier elimination on termination.
At first we also consider solutions to WSOQE that are not necessarily first-order and give them a name: A tuple of predicate expressions $\seq{\witness}$ is called an \emph{elimination witness}, or simply a \emph{witness for $\Exists{\seq{X}} \varphi$}, if it satisfies the WSOQE-condition \eqref{wsoqe-condition} for $\Exists{\seq{X}} \varphi$.

An illustration of how WSCAN works is given in the following diagram:
$$
\begin{array}{rcccccccl}
  N = N(D)_0
  &\xrightarrow{D_1}& N(D)_1
  &\xrightarrow{D_2}& \quad \dots \quad&
  N(D)_{m-1} &\xrightarrow{D_m}& \ N(D)_m \\[0.3em]
  \seq{\mathrm{wit}(D)} = \seq{\mathrm{wit}_{0}(D)}
  & \xleftarrow{T_{D_1}}& \seq{\mathrm{wit}_{1}(D)}
  & \xleftarrow{T_{D_2}}& \quad \dots \quad & \seq{\mathrm{wit}_{m-1}(D)} &\xleftarrow{T_{D_m}} 
  &\ \seq{\mathrm{wit}_m(D)}
\end{array}
$$
We start by saturating a given clause set $N$ using the derivation steps of the calculus and collect them in a derivation $D = (D_i)_{1 \leq i \leq m}$.
If saturation terminates we construct a witness $\seq{\mathrm{wit}(D)}$ for $\Exists{\seq{X}} N$.
To this aim we iteratively compute witnesses $\seq{\mathrm{wit}_i(D)}$ for $\Exists{\seq{X}} N(D)_i$ for all $0 \leq i \leq m$ starting from the back with~$\seq{\mathrm{wit}_m(D)}$.
Since $D$ resulted from saturating the input clause set $N$ we have that $N(D)_m$, i.e.,  the conclusion of $D$ from $N$, contains none of the variables in $\seq{X}$, otherwise a $\mathcal{C}$-derivation step would be applicable.
We now need an initial witness $\seq{\mathrm{wit}_m(D)}$ for $N(D)_m$.
Since $N(D)_m$ does not contain any of the variables of $\seq{X}$ choosing for $\seq{\mathrm{wit}_m(D)}$ any tuple of predicate expressions (compatible with the arities of $\seq{X}$) produces a witness.
We then iteratively transform witness $\seq{\mathrm{wit}_i(D)}$ of $N(D)_i$ to a witness $\seq{\mathrm{wit}_{i-1}(D)}$ of $N(D)_{i-1}$ using an operation $T_{D_i}$.
Finally, the tuple of predicates $\seq{\mathrm{wit}_0(D)}$ is a witness for $\Exists{\seq{X}} N$.

We now want to define the transformation $T_S$ for the different derivation steps $S$ in $\mathcal{C}$.
It has to satisfy the following property:
\begin{restatable}[Witness preservation lemma]{lemma}{witnessPreservationLemma}
  \label{witness-preservation-lemma}
  \ If $S$ is a $\mathcal{C}$-derivation step from~$N$ to $N'$ and $\seq{\witness}$ is a witness for $\Exists{\seq{X}} N'$,
  then $T_S(\seq{\witness})$ is a witness for $\Exists{\seq{X}} N$.
\end{restatable}

Before we give the formal definition of $T_S$, we give some intuition on how to satisfy \Cref{witness-preservation-lemma}.
For $S \in \mset{\mathrm{Res}, \mathrm{Fac}, \mathrm{ConstrElim}, \mathrm{RedElim}}$ we have $N \iff N'$.
Therefore we can use the same witness, i.e., we set $T_S(\seq{\witness})~=~\seq{\witness}$.
If we have $S = \mathrm{ExtPurDel}_X^{+}$ we see that $N' \iff N[X \leftarrow \lambda \seq{u}. \top]$ since every clause in $N$ that contains $X$ has a positive occurrence of $X$, making those clauses logically equivalent to~$\top$ under $[X \leftarrow \lambda \seq{u}. \top]$.
If $\seq{\witness}$ is a witness for $\Exists{\seq{X}} N'$,
then setting the $X$-component of $\seq{\witness}$ to $\lambda \seq{u}. \top$ and leaving the other components the same results in a witness for $N$.
Similarly, if $S = \mathrm{ExtPurDel}_X^{-}$ we set the $X$-component of $T_S(\seq{\witness})$ to $X$ to $\lambda \seq{u}. \bot$. Then $N' \iff N[X \leftarrow \lambda \seq{u}. \bot]$.

Defining the transformation $T_S$ for $S = \mathrm{PurDel}_P$ needs a bit of groundwork.
A central piece is the resolution closure of a clause set with respect to a pointed clause:

\begin{definition}
  Let $P$ be a pointed clause with designated literal $L$ and let $N$ be a clause set.
  We define $\resclosure{P}{N}$ to be the smallest set such that
  (i) $N \subseteq \resclosure{P}{N}$, and (ii)\label{def.res.res} if $C \in \resclosure{P}{N}$ and $R$ is a resolvent between $C$ and $P$ upon $L$, then $R \in \resclosure{P}{N}$.
\end{definition}
For a clause $C$ we define the abbreviation $\resclosure{P}{C}~\definedas~\resclosure{P}{\mset{C}}$
and for a literal $L$ we set $\resclosure{P}{L}~\definedas~\resclosure{P}{\mset{L}}$.

Furthermore for $P = \underline{L(\seq{t})} \lor C_P$ we define the set $\resclosurepremise{P}{\seq{c}}~\definedas~\resclosure{P}{L(\seq{c})^\perp}$, where $\seq{c}$ is a tuple of fresh constants.
The `$\mathrm{U}$' stands for unit and signifies that we consider the resolution closure with respect to $P$ starting with a unit clause dual to the designated literal of $P$.

To simplify the presentation of the set $\resclosure{P}{N}$ we often tacitly replace it by a logically equivalent clause set, for example by replacing a clause $C \in \resclosure{P}{N}$ by a logically equivalent clause $C'$, which results from $C$ by constraint elimination.

Some examples of $\resclosurepremise{P}{\seq{c}}$ are given in the following:
\begin{example}
  \label{ex.resolution-closure}
  Let $P = \underline{X(a)}$, then we have
  $\resclosurepremise{P}{c} = \mset{\neg X(c), a \noeq c}$ and for $P' = \underline{\neg X(a)} \lor \neg X(b)$ we get
  $\resclosurepremise{P'}{c} = \mset{X(c), a \noeq c \lor \neg X(b)}$.
  
  For $P'' = B(u,v) \lor \underline{\neg X(u)} \lor X(v)$ we get an infinite set: 
  \begin{align*}
    \resclosurepremise{P''}{c} = \{
    &X(c),\ B(c,v) \lor X(v), \\
    &B(c,v) \lor B(v,v') \lor X(v'), \\
    &B(c,v) \lor B(v,v') \lor B(v',v'') \lor X(v''), \dots \}.
  \end{align*}
\end{example}

In the definition of the clause set $\resclosurepremise{P}{\seq{c}}$ we use a tuple of constants $\seq{c}$ instead of variables as we want to avoid the implicit universal quantification over the clauses for these placeholders.
These constants are intended as arguments for a predicate expression $\resclosurepredicate{P}$ (`$\mathrm{p}$' for predicate) given by:
\begin{definition}
  \label{res-neg-definition}
  Let $P$ be an $X$-pointed clause. We set
  $$
    \resclosurepredicate{P} \definedas
    \begin{cases}
      \lambda \seq{u}. \Land_{C'(\seq{c}, \seq{v}) \in \resclosurepremise{P}{\seq{c}}} \Forall{\seq{v}} C'(\seq{u}, \seq{v})     & \text{ if } P = \underline{\neg X(\seq{t})} \lor C \text{ and} \\
      \lambda \seq{u}. \Lor_{C'(\seq{c}, \seq{v}) \in \resclosurepremise{P}{\seq{c}}} \Exists{\seq{v}} \neg C'(\seq{u}, \seq{v}) & \text{ if } P = \underline{X(\seq{t})} \lor C,
    \end{cases}
  $$
  where $\seq{u}$ is a tuple of fresh variables of the same size as $\seq{c}$.
\end{definition}

\begin{example}
  For $P = \underline{X(a)}$ and $P' = \underline{\neg X(a)} \lor \neg X(b)$ as in \Cref{ex.resolution-closure} we get $\resclosurepredicate{P} = \lambda u. X(u) \lor u \oeq a$ and $\resclosurepredicate{P'} = \lambda u. X(u) \land (u \noeq a \lor \neg X(b))$.
  If $\resclosurepremise{P}{\seq{c}}$ is infinite, as for $P''$ of \Cref{ex.resolution-closure}, then $\resclosurepredicate{P}$ is an infinite predicate expression.
\end{example}

A crucial property of $\resclosurepredicate{P}$ is the following:
\begin{restatable}{lemma}{resclosurepredicateLemma}
  \label{resclosurepredicate-lemma}
  Let $P$ be an $X$-pointed clause and let $C$ be a clause. 
  Then we have $\models P[X \leftarrow \resclosurepredicate{P}]$ and $\models \resclosure{P}{C} \limp C[X \leftarrow \resclosurepredicate{P}]$.
\end{restatable}

A way to view this lemma is (i) $X \leftarrow \resclosurepredicate{P}$ is always a satisfying assignment for $P$ and (ii) $X \leftarrow \resclosurepredicate{P}$ is a satisfying assignment for $C$, given that there is an assignment of $X$ that satisfies all clauses in $\resclosure{P}{C}$ (note that $X$ can occur freely in $\resclosure{P}{C}$ and $\resclosurepredicate{P}$).
The reader might gain an intuition for the lemma by applying it to the pointed clauses from the previous example.

Using \Cref{resclosurepredicate-lemma} we can show the witness preservation lemma (\Cref{witness-preservation-lemma}) for purified clause deletion steps using $\resclosurepredicate{P}$:
\begin{restatable}{lemma}{witnessPreservationLemmaPurDel}
  \label{witness-preservation-lemma.purdel}
  Let $N$ be a clause set, let $P$ be $X_i$-pointed and purified in $N$ and suppose $\seq{\witness} = (\witness_1, \dots, \witness_d)$ is a witness for $\Exists{\seq{X}} N$.
  Then a witness for $\Exists{\seq{X}}(N \union \mset{P})$ is given by $(\witness_1, \dots, \witness_{i-1}, \resclosurepredicate{P}[\seq{X} \leftarrow \seq{\witness}], \witness_{i+1}, \dots, \witness_d)$.
\end{restatable}
\begin{proof}[Proof sketch]
  To avoid notational overhead we give a proof sketch under the simplifying assumption $d=1$, i.e., we only consider a single variable $X$.

  To show that $\witness_P := \resclosurepredicate{P}[X \leftarrow \witness]$ is a witness for $N_P := N \union \mset{P}$ we need to show that both directions of the WSOQE-condition \eqref{wsoqe-condition} hold.
  The direction $\models N_P[X \leftarrow \witness_P] \limp \Exists{X} N_P$ holds as $\witness_P$ is an instance of the existential quantifier.
  Now it suffices to show $\models \Exists{X} N_P \limp N_P[X \leftarrow \witness_P]$.
  Since $N \subseteq N_P$ we have $\models \Exists{X} N_P \limp \Exists{X} N$
  and because $\witness$ is a witness for $\Exists{X} N$ we have $\models \Exists{X} N \limp N[X \leftarrow \witness]$.
  It remains to show $\models N[X \leftarrow \witness] \limp N_P[X \leftarrow \witness_P]$.
  We reduce this to $\models P[X \leftarrow \witness_P]$ and $\models N[X \leftarrow \witness] \limp C[X \leftarrow \witness_P]$ for all $C \in N$.
  
  Using the first statement of \Cref{resclosurepredicate-lemma} we get $\models P[X \leftarrow \resclosurepredicate{P}]$.
  Applying the substitution which sends $X$ to $\witness$ results in $\models P[X \leftarrow \resclosurepredicate{P}][X \leftarrow \witness]$.
  Note that applying $[X \leftarrow \witness]$ after $[X \leftarrow \resclosurepredicate{P}]$ is the same as applying $[X \leftarrow \resclosurepredicate{P}[X \leftarrow \witness]]$.
  Thus $\models P[X \leftarrow \witness_P]$.

  To show $\models N[X \leftarrow \witness] \limp C[X \leftarrow \witness_P]$ we use that $P$ is purified in $N$ to get that all clauses in $\resclosure{P}{N}$, 
  and in particular all clauses in $\resclosure{P}{C} \subseteq \resclosure{P}{N}$, are redundant in $N$.
  Therefore $\models N \limp \resclosure{P}{C}$ and we can apply the second statement of \Cref{resclosurepredicate-lemma} to get $\models N \limp C[X \leftarrow \resclosurepredicate{P}]$.
  Again, applying $[X \leftarrow \witness]$ results in the desired statement $\models N[X \leftarrow \witness] \limp C[X \leftarrow \witness_P]$.
\end{proof}

We now give the formal definition of the witness transformation:
\begin{restatable}{definition}{witnessTransformationDefinition}
  \label{ext-definition}
  Let $S$ be a $\mathcal{C}$-derivation step and let $\seq{X} = (X_1, \dots, X_d)$ be a tuple of predicate variables.
  For $1 \leq i \leq d$, let $\seq{u_i}$ be a tuple of distinct variables, the size of which is the arity of $X_i$.
  Also, let $\seq{\witness} = (\witness_1, \dots, \witness_d)$ be a tuple of predicate expressions such that $\witness_i$ has the same arity as $X_i$ for $1 \leq i \leq d$. 
  We define the predicate expression $T_S(\seq{\witness})$ by a case distinction on $S$ as follows:
  $$
  \begin{array}{rll}
    T_{S}(\seq{\witness}) &:= &\seq{\witness} \text{ for $S \in \mset{\mathrm{Res}, \mathrm{Fac}, \mathrm{ConstrElim}, \mathrm{RedElim}}$}, \\
    T_{\mathrm{ExtPurDel}^{+}_{X_i}}(\seq{\witness}) &:= &(\witness_1, \dots, \witness_{i-1}, \lambda \seq{u_i}. \top, \witness_{i+1}, \dots, \witness_d),\\
    T_{\mathrm{ExtPurDel}^{-}_{X_i}}(\seq{\witness}) &:= &(\witness_1, \dots, \witness_{i-1}, \lambda \seq{u_i}. \bot, \witness_{i+1}, \dots, \witness_d), \text{ and}\\
    T_{\mathrm{PurDel}_P}(\seq{\witness}) &:= &(\witness_1, \dots, \witness_{i-1}, \resclosurepredicate{P}[\seq{X} \leftarrow \seq{\witness}], \witness_{i+1}, \dots, \witness_d), \\
    & &\text{if $P$ is $X_i$-pointed}.
  \end{array}
  $$
\end{restatable}
Above we explained why the witness preservation lemma holds for these $T_S(\seq{\witness})$.
We now construct witnesses iteratively given a $\mathcal{C}$-derivation $(D_i)_{1 \leq i \leq m}$ starting from the back:
\begin{restatable}{definition}{derivationWitnessDefinition}
  \label{witness-definition}
  Let $N(X_1, \dots, X_d)$ be a clause set and let $D = (D_i)_{1 \leq i \leq m}$ be a $\mathcal{C}$-derivation from $N$. 
  For $1 \leq i \leq d$ let $W_i$ be a predicate variable with the same arity as $X_i$.
  We inductively define a tuple $\seq{\mathrm{wit}_i(D)}$ of predicate expressions, for $0 \leq i \leq m$, starting with $\seq{\mathrm{wit}_m(D)}$ as the base case by
  $$
  \begin{aligned}
    \seq{\mathrm{wit}_{m}(D)} &:= (W_1, \dots, W_d) 
    \quad \text{and} \quad \seq{\mathrm{wit}_{i-1}(D)} &:= T_{D_i}(\seq{\mathrm{wit}_i(D)}).
  \end{aligned}
  $$
  Finally, the witness for $\Exists{\seq{X}} N$ is given by $\seq{\mathrm{wit}(D)} := \seq{\mathrm{wit}_0(D)}$.
\end{restatable}

If $D = (D_i)_{1 \leq i \leq m}$ is an $\seq{X}$-eliminating $\mathcal{C}$-derivation from $N$, 
the conclusion of $D$ does not contain any of the predicate variables in $\seq{X}$.
Thus, any tuple of predicate expressions (compatible with the arities in $\seq{X}$) provides a witness for the conclusion of $D$. 
The predicate variables $W_i$ in the above definition serve as placeholders for this arbitrary choice of witness and any instantiation of these predicate variables provides a witness.
We call $W_i$ the \emph{witness parameter corresponding to $X_i$}. 
The $\seq{\mathrm{wit}_i(D)}$ are witnesses for the intermediate clause sets $N(D)_i$ which can be proved by an inductive argument on the length of $D$ using \Cref{witness-preservation-lemma}.

Our main result is:
\begin{restatable}{theorem}{mainresult}
  \label{witnesses-from-eliminating-scan-derivations}
  If $D$ is an $\seq{X}$-eliminating $\mathcal{C}$-derivation from $N$, then $\seq{\mathrm{wit}(D)}$ is a witness for $\Exists{\seq{X}} N$.
\end{restatable}

Here is an example illustrating the algorithm:
\begin{example}
  \label{ex.main-example}
  Remember the clause set $N$ from \Cref{ex.main-example.derivation}:
  \begin{align*}
    (1)\ B(a,v) \qquad (2)\ X(a) \qquad (3)\ B(u,v) \lor \neg X(u) \lor X(v) \qquad  (4)\ \neg X(c)
  \end{align*}
  and suppose we want to find a witness for $X$.
  Some resolvents from $N$ are
  \begin{align*}
    (5)\ a \noeq u \lor B(u,v) \lor X(v) \quad \text{($2$ with $3$)} \quad \text{and} \quad
    (6)\ & a \noeq c & \text{($2$ with $4$)}.
  \end{align*}
  Also recall the derivation $D$:
  $$
    \Axiom$\mset{1,2,3,4}\fCenter$
    \RightLabel{$\mathrm{Res}_{2.1,4.1}$}
    \UnaryInf$\mset{1,2,3,4,6}\fCenter$
    \RightLabel{$\mathrm{PurDel}_{2.1}$}
    \UnaryInf$\mset{1,3,4,6}\fCenter$
    \RightLabel{$\mathrm{ExtPurDel}_X^{-}$}
    \UnaryInf$\mset{1,6}\fCenter$
    \DisplayProof
  $$

  In computing $\mathrm{wit}_i(D)$,
  we freely use equivalences of the form $\bot \lor \varphi \iff \varphi$, $\top \land \varphi \iff \varphi$ and $\neg \neg \varphi \iff \varphi$ as well as commutativity and associativity of $\land$ and $\lor$ to simplify the constructed witnesses.
  We use $W_X$ to denote the witness parameter corresponding to $X$.
  The steps to construct $\mathrm{wit}(D)$ are:
  \renewcommand{\arraystretch}{1.0}
  $$
  \begin{array}{r@{\quad}l@{\quad}l@{\quad}l@{\quad}l}
    i\ & N(D)_i & D_{i+1} & T_{D_{i+1}}(\witness) &  \mathrm{wit}_{i}(D) \\
    \hline
    3\ & \mset{1,6} & - & - & W_X \\
    2\ & \mset{1,3,4,6} & \mathrm{ExtPurDel}_X^{-} & \lambda u. \bot & \lambda u. \bot  \\
    1\ & \mset{1,2,3,4,6} & \mathrm{PurDel}_{2.1} &  \resclosurepredicate{2.1}[X \leftarrow \witness] & \lambda u. u \oeq a \\
    0\ & \mset{1,2,3,4} & \mathrm{Res}_{2.1,4.1} & \witness & \lambda u. u \oeq a
  \end{array}
  $$
  \renewcommand{\arraystretch}{1.0}%
  Note that $\resclosurepremise{2.1}{d} = \mset{\neg X(d), d \noeq a}$ and thus $\resclosurepredicate{2.1} = \lambda u. X(u) \lor u \oeq a$.
  Therefore $\mathrm{wit}(D) = \lambda u. u \oeq a$ is a WSOQE-witness for $\Exists{X} N$.
\end{example}

\section{First-order Witnesses}
\label{sec.first-order-witnesses}

At this point it is not clear that the components of $\seq{\mathrm{wit}(D)}$ are always first-order predicates as the expression $\resclosurepredicate{P}$ can be infinite.
In fact, there are $\seq{X}$-eliminating derivations where the predicate expressions in $\seq{\mathrm{wit}(D)}$ are not equivalent to first-order predicates.
This can be avoided, if for all purified clauses~$P$ in $D$ there are first-order predicates $F_P$ equivalent to $\resclosurepredicate{P}$:
\begin{definition}
  Let $D$ be a $\mathcal{C}$-derivation from $N$.
  Let $F$ be a mapping that assigns to all purified pointed clauses $P$ of $D$ a first-order predicate $F_P$ with $\resclosurepredicate{P} \iff F_P$.
  Then $F$ is called a \emph{first-order annotation for $D$}.
  In that case, denote by $\seq{\mathrm{wit}(D,F)}$ the tuple of first-order predicate expressions resulting from $\seq{\mathrm{wit}(D)}$ where every occurrence of $\resclosurepredicate{P}$ is replaced by $F_P$. 
\end{definition}

Now the first-order version of our main result is:
\begin{restatable}{corollary}{finiteMainResult}
  \label{finite-witnesses-from-eliminating-scan-derivations}
  If $D$ is an $\seq{X}$-eliminating $\mathcal{C}$-derivation from $N$ and $F$ is a first-order annotation for $D$,
  then $\seq{\mathrm{wit}(D,F)}$ is a first-order witness for $\Exists{\seq{X}} N$.
\end{restatable}

To find such a first-order annotation $F$ we compute for all purified pointed clauses $P$ a finite clause set $N_P$ logically equivalent to $\resclosurepremise{P}{\seq{c}}$.
Then replacing $\resclosurepremise{P}{\seq{c}}$ by $N_P$ in the definition of $\resclosurepredicate{P}$ results in a first-order predicate expression $F_P$ equivalent to $\resclosurepredicate{P}$.
We use SCAN's purification process to find such finite clause sets $N_P$:
Let $P = \underline{L(\seq{t})} \lor C$.
Start with the clause set $\mset{L(\seq{c})^\perp}$ and saturate it with respect to constraint resolution with $P$.
In addition we use constraint and redundancy elimination to simplify the intermediate clause sets.
If at some point all constraint resolvents of the clause set with $P$ are redundant, we obtain a finite clause set $N_P$ which is logically equivalent to $\resclosurepremise{P}{\seq{c}}$.
However, this process might not terminate as is the case for $P''$ from \Cref{ex.resolution-closure}.
Thus it is unclear whether computing a first-order annotation $F$ is always possible.

To get a concrete strategy for constructing $\mathcal{C}$-derivations where \Cref{finite-witnesses-from-eliminating-scan-derivations} is applicable we direct our search to a decidable property which is sufficient to guarantee first-order annotations.
One such property is:
\begin{definition}
  An $X$-pointed clause $P$ is called \emph{one-sided} if $X$ occurs only positively or only negatively in $P$.
  A $\mathcal{C}$-derivation $D$ is called \emph{one-sided} if all purified pointed clauses in $D$ are one-sided (not necessarily with the same polarity).
\end{definition}

Derivation $D$ in \Cref{ex.main-example} is one-sided as its only purified pointed clause~$2.1$ is one-sided.
In general, if $P = \underline{L(\seq{t})} \lor C$ is one-sided,
then we have $\resclosurepremise{P}{\seq{c}} = \mset{L(\seq{c})^\perp, \seq{t} \noeq \seq{c} \lor C}$, i.e., $\resclosurepremise{P}{\seq{c}}$ is finite.
Thus $P \mapsto \resclosurepredicate{P}$ is a first-order annotation for one-sided $\mathcal{C}$-derivations.
One-sidedness of a pointed clause $P$ is not necessary for the existence of a first-order predicate equivalent to $\resclosurepredicate{P}$, e.g., for $P = \underline{\neg X(u,v)} \lor X(v,u)$ we have $\resclosurepremise{P}{c,d} \iff \mset{X(c,d), X(d,c)}$, therefore $F_P = \lambda u \lambda v. X(u,v) \land X(v,u)$ satisfies $F_P \iff \resclosurepredicate{P}$.

As a special case of \Cref{finite-witnesses-from-eliminating-scan-derivations} we get:
\begin{restatable}{corollary}{oneSidedMainResult}
  \label{finite-witnesses-from-one-sided-scan-derivations}
  If $D$ is a one-sided $\seq{X}$-eliminating $\mathcal{C}$-derivation from $N$,
  then $\seq{\mathrm{wit}(D)}$ is a first-order witness for $\Exists{\seq{X}} N$.
\end{restatable}

Since one-sidedness is a decidable property, this suggests a strategy to construct $\seq{X}$-eliminating $\mathcal{C}$-derivations which induce WSOQE-witnesses: Run SCAN in such a way that only one-sided pointed clauses $P$ are purified.
If SCAN terminates, the resulting derivation will also be one-sided, ensuring a first-order witness $\seq{\mathrm{wit}(D)}$.
However, it is open whether this affects completeness of the calculus (see Problem \ref{completeness-of-onesided-derivations}).

Furthermore for one-sided $\seq{X}$-eliminating derivations $D$ one can show a tight exponential upper bound on the size of $\seq{\mathrm{wit}(D)}$ in terms of the length of $D$.
If every purified pointed clause in $D$ contains at most one $\seq{X}$-literal, the bound is even linear.

\section{Implementation}
\label{sec.implementation}

The presented method is implemented in the Scala programming language~\cite{scala} as part of the General Architecture for Proof Theory (GAPT) software package~\cite{GAPT} which provides useful data structures to develop algorithms in computational logic, including terms, formulas and clauses.
The prototype is available in version 2.18.1 of GAPT (\texttt{https://www.logic.at/gapt/release\_archive.html}).
Section 8.4 of the GAPT user manual details how to use the prototype.

To evaluate our witness construction method we have implemented the SCAN algorithm which, in addition to finding a logically equivalent first-order clause set, provides a $\mathcal{C}$-derivation which is then used by our method to extract a witness.

The saturation loop proceeds by picking a pointed clause $P = \underline{L(\seq{t})} \lor C$ from the active clause set and performs purification, i.e., adds all non-redundant resolvents between $P$ and the rest of the clause set while eagerly applying constraint elimination, tautology deletion, subsumption deletion, as well as inferring all non-redundant constraint factors of newly added clauses.
Once $P$ is purified, it is deleted from the active clause set.
Since this process might not terminate we provide parameters that limit the inferences to be performed.
If the input clause set is not saturated within this limit the procedure reports an error.
Since one-sided derivations guarantee first-order witnesses, we also provide an option to only apply the purification process to one-sided pointed clauses.
During the saturation process all performed steps are stored in a $\mathcal{C}$-derivation, modelled as the initial clause set together with a list of $\mathcal{C}$-derivation steps.

The choice of pointed clause determines the outcome of the saturation process.
Different choices can lead to termination or non-termination.
This also affects which $\mathcal{C}$-derivations are found.
To find all possible $\mathcal{C}$-derivations from a given clause set, 
we have implemented backtracking on top of the saturation process to be able to get $\mathcal{C}$-derivations based on the different choices of pointed clauses.
This backtracking can be performed on-demand, and allows to find multiple derivations and therefore multiple witnesses.

The witness extraction method recurses on the produced $\mathcal{C}$-derivation following \Cref{witness-definition}.
We implemented $T_S$ according to \Cref{ext-definition} where $\resclosurepredicate{P}$ is computed by using the same purification process as described above, but starting with the clause set consisting of the single literal $L(\seq{c})^\perp$.
This means we also use redundancy criteria and constraint factoring inferences to simplify the clause set $\resclosurepremise{P}{\seq{c}}$.
Since $\resclosurepredicate{P}$ is potentially infinite we also provide parameters that limit the inference steps performed in its computation.
If the witness computation succeeds, it returns a substitution that maps the eliminated symbols to the computed first-order witness predicates.

Our prototype implementation has been tested on 26 examples from the file
\texttt{examples/predicateEliminationProblems.scala} included in GAPT 2.18.1.
These examples were either created by us or selected from the literature and
possibly modified. Tests were run using the GAPT shell with the Azul Zulu
Java JDK 21.0.2 on a MacBook Pro with an M2 Pro CPU and 32GB RAM. We run the \texttt{wscan} method with default parameters, measuring execution
time, input clause set size, and resulting witness size (if found).
We measure the size of a literal by the number of non-logical symbols it contains.
The size of a clause is given by the sum of the sizes of its literals and the size of a clause set is given by the sum of the sizes of its clauses.
The size of a witness is measured by the number of logical and non-logical symbols it contains.
See the \texttt{wscanTest} method in the file \texttt{examples/predicateEliminationProblems.scala} for the exact test protocol. 

The results are the following: The input size ranged from 2 to 177 with an average of 21.58.
For 21 of the 26 examples a witness could be computed.
In the remaining five cases a user-set timeout based on the number of derivation steps was reached
 without finding a derivation.
For the 21 terminating examples the running times ranged from 0.03 milliseconds to 150.60
 milliseconds with an average of 14.96 milliseconds.
The witness sizes ranged from 0 (this example had no second-order quantifiers) to 220 with an average of 21.50.

We have not optimized our implementation for large clause sets.
Furthermore, we are not aware of a sufficiently large database of examples to test our method on.
This means a thorough empirical evaluation is left as future work.

\section{Discussion and Future Work}\label{sec.discussion}

\subsubsection{Constructing multiple witnesses.}
Different $\mathcal{C}$-derivations can result in different, non-equivalent witnesses. Thus one can use our method to construct multiple witnesses for a given input clause set.
As already mentioned, to get different $\mathcal{C}$-derivations our prototype performs backtracking on different choices of pointed clauses in the purification process of SCAN.

Another approach to finding multiple witnesses would be to instantiate the free predicate variables $W_1, \dots, W_d$ in the definition of $\mathrm{wit}(D)$.
These are contained in the witness if $D$ does not use extended purity deletion.

In practice certain witnesses might be preferred over others, e.g., based on size.
Accounting for additional constraints like these when constructing witnesses is an avenue for future work.

\subsubsection{Limitations for finding witnesses.}
\label{sec.skolemization}
The original SCAN algorithm uses Sko\-lem\-ization if the input formula is not in clause form and after the saturation process tries to undo the Skolemization to find a logically equivalent first-order formula.
While reverse Skolemization is undecidable in general, it does work in some cases, for example, if the input formula has the form $\Exists{\seq{u}} \varphi$ where $\varphi$ is quantifier-free.
We discuss such an example where the SCAN algorithm terminates and finds a SOQE-solution, but where no WSOQE-witness exists.

The input formula is $\Phi = \Exists{X} \Exists{u} \Exists{v} ( X(u) \land \neg X(v))$.
Skolemization of the first-order part yields the clause set $N = \mset{X(a), \neg X(b)}$ with Skolem constants~$a$ and~$b$.
An $X$-eliminating $\mathcal{C}$-derivation from $N$ is given by $D = $
$$
  \Axiom$\mset{X(a), \neg X(b)}\fCenter$
  \RightLabel{$\mathrm{Res}_{\underline{X(a)}, \underline{\neg X(b)}}$}
  \UnaryInf$\mset{X(a), \neg X(b), a \noeq b}\fCenter$
  \RightLabel{$\mathrm{PurDel}_{\underline{X(a)}}$}
  \UnaryInf$\mset{\neg X(b), a \noeq b}\fCenter$
  \RightLabel{$\mathrm{ExtPurDel}_X^{-}$}
  \UnaryInf$\mset{a \noeq b}\fCenter$
  \DisplayProof
$$
The original SCAN algorithm now reverses the Skolemization on the resulting formula $a \noeq b$, to get the logically equivalent formula $\Exists{u} \Exists {v} u \noeq v$, i.e., we have $\Phi \iff \Exists{u} \Exists {v} u \noeq v$.
The computed witness $\mathrm{wit}(D)=\resclosurepredicate{\underline{X(a)}}[X \leftarrow \lambda u.\bot]$ is equivalent to $\lambda u. u \oeq a$.
However, this witness contains a Skolem constant which is not in the language of the input formula.
In fact, no WSOQE-witness exists in the input language:
Assume $\Phi$ had a WSOQE-witness $\witness = \lambda u. \varphi(u)$.
Then $\witness$ is a first-order predicate in the empty language and in every structure $\mathcal{M}$ it defines the subset $\witness^{\mathcal{M}}$ given by $\mset{m \in M \suchthat \mathcal{M} \models \varphi(u)}$.
For every two-element model $\mathcal{M} = \mset{m_1, m_2}$ 
we have $\mathcal{M} \models \Exists{X} \Exists{u} \Exists{v} (X(u) \land \neg X(v))$.
Since $\witness$ is a witness we get
$\mathcal{M} \models \Exists{u} \Exists{v}(\varphi(u) \land \neg \varphi(v))$.
Thus $\witness^{\mathcal{M}}$ is neither the empty set nor the whole set $M$.
However, in the empty language the only first-order parameter-free definable sets in a structure
 $\mathcal{M}$ are the empty set and $M$ itself, see~\cite[Theorem 2.1.2]{Hodges97Shorter}, which is a contradiction.

This shows a fundamental limitation for finding witnesses, even if SOQE-solutions exist.
For practical applications it would be useful to add Skolemization and reverse Skolemization steps to the implementation and investigate classes of formulas where witness construction succeeds in the presence of Skolemization.

\subsubsection{Quantifier alternations}

We have discussed the WSOQE problem in the context of existential quantifiers, but one can also consider the dual problem for universal quantifiers:
Given a formula $\Forall{\seq{X}} \varphi$ with first-order~$\varphi$, find first-order predicates $\seq{\witness}$ such that
\begin{equation*}
  \Forall{\seq{X}} \varphi \iff \varphi[\seq{X} \leftarrow \seq{\witness}].
\end{equation*}
One can reduce this problem to the existential case since for all first-order predicates $\seq{\witness}$ we have $\Forall{\seq{X}} \varphi \iff \varphi[\seq{X} \leftarrow \seq{\witness}]$ if and only if $\Exists{\seq{X}} \neg \varphi \iff \neg \varphi[\seq{X} \leftarrow \seq{\witness}]$ by writing $\Forall{\seq{X}}$ as $\neg \Exists{\seq{X}} \neg$.
This way, one can solve the problem for an arbitrary quantifier prefix by successively eliminating alternating blocks of existential and universal quantifiers starting with the innermost block.

While the reduction above works if one has a general WSOQE method that can handle any input formula, 
our submission only
works when the input is a clause set, 
i.e., the domain variables are universally
quantified.
The above reduction introduces a negation on the input formula, 
meaning that we now have to deal with existential quantifiers on
the domain variables as well, e.g., via Skolemization. 
However, as stated before there are fundamental limitations when introducing Skolemization.
For input formulas where the first-order part is quantifier-free our method is applicable though.

\subsubsection{Completeness of one-sided $\seq{X}$-eliminating derivations.}

To ensure that SCAN produces a one-sided derivation, if terminating, one can restrict the application of $\mathrm{PurDel}_P$ to one-sided $P$.
However, this might restrict the number of clause sets for which an $\seq{X}$-eliminating derivation can be found.
It is open whether one-sided $\seq{X}$-eliminating derivations are complete in this sense:
\begin{problem}
  \label{completeness-of-onesided-derivations}
  If there is an $\seq{X}$-eliminating $\mathcal{C}$-derivation from $N$, is there always
  a one-sided $\seq{X}$-eliminating $\mathcal{C}$-derivation from $N$?
\end{problem}
A positive answer would justify using the strategy of restricting $\mathrm{PurDel}_P$ to one-sided $P$ during a run of the SCAN algorithm which, if terminating, guarantees a first-order witness by \Cref{finite-witnesses-from-one-sided-scan-derivations}.

\subsubsection{Computing witnesses with DLS(*).}
As this work extends a known algorithm for solving SOQE to solve the more general WSOQE problem, the question arises whether other algorithms like DLS \cite{Doherty97Computing} and DLS* \cite{Doherty98General,Nonnengart98Fixpoint} can be equally modified to solve WSOQE and how the resulting witnesses compare.

\subsubsection{Improvement of Ackermann's Lemma.}
\label{sec.discussion-ackermann-dls}

Ackermann's Lemma provides a method for solving WSOQE as well.
On clause sets, our method is an improvement on Ackermann's Lemma in the following sense:
If Ackermann's Lemma applies to $\Exists{X} N$, then there is a one-sided $\mathcal{C}$-derivation $D$ from $N$, such that $\mathrm{wit}(D)$ is equivalent to the witness produced by Ackermann's Lemma.
Furthermore our method can construct a witness for \Cref{ex.main-example} where Ackermann's Lemma is not applicable.
This result makes no claim about whether our method improves on other Ackermann-based algorithms, such as DLS and DLS*, since they employ other transformation steps besides the application of Ackermann's Lemma.

\subsubsection{Equality reasoning.}
A desirable extension is an implementation with
equality reasoning.
The original SCAN algorithm works in the presence of
equality~\cite{Bachmair94Refutational}, so we expect the witness construction to carry over
to first-order logic with equality.

\subsubsection{Solving FEQ.}

As discussed in the Introduction, any WSOQE algorithm can be applied to solving formula equations, by appending a validity check with a first-order theorem prover.
In some cases our method can skip the additional validity check: If SCAN terminates with the empty set (which implies that $\Exists{\seq{X}} N$ is valid) and the resulting derivation is one-sided, the witness found is also a solution to the FEQ problem.
However, SCAN might output a valid clause set which is not the empty set, in which case a validity check is still required to solve FEQ.

This opens the door to apply this method in Boolean unification (which is a restriction of FEQ to nullary predicate variables) and in the context of software verification for solving sets of constrained Horn clauses~\cite{Bjorner15Horn}.

\subsubsection{Application to Forgetting.}
Methods for computing forgetting solutions and uniform
interpolants have been developed for description logics and modal
logics~\cite{KoopmannSchmidt13a,KoopmannSchmidt13c,KoopmannSchmidt15a,AlassafSchmidtSattler22}.
For applications, e.g., knowledge processing, agent applications and
Boolean unification, the possibility of extending these to compute also
witnesses for eliminated variables would be attractive.
Using a WSOQE-solution instead of a SOQE-solution has the advantage that the structure of the
existing knowledge base is left largely intact.
Functionality to return witnesses could
provide a promising avenue to develop methods for such applications, but
such development is subject to future work.

\subsubsection{Expressivity.}
Investigating SOQE, and related problems, e.g., uniform
interpolation, one quickly encounters expressibility challenges.
These carry over to WSOQE, as we have seen in this paper, but also to FEQ
and Boolean unification as a special instance of FEQ.
However, applications might be able to cope with language extensions, e.g., 
the extent to which infinite witnesses can be expressed using fixpoints is worth investigating. 
Enhancing the expressivity of the target language might increase success rates.

\section{Conclusion}

We introduced WSCAN, an extension of the SCAN algorithm on clause sets, to compute witnesses for existential second-order quantifiers, thereby solving the more general WSOQE-problem.
Given an $\seq{X}$-eliminating derivation $D$ from a clause set $N$ we can produce a witness for $\Exists{\seq{X}} N$.
That witness might be infinite in general, but if in addition all purified pointed clauses of $D$ have a certain finiteness property, like one-sidedness, we can guarantee that the witness is first-order.
To the best of our knowledge there are currently no other WSOQE-algorithms applicable to arbitrary clause sets.
This work paves the way to new applications of the SCAN algorithm in a variety of different areas,
ranging from modal correspondence theory and knowledge representation to verification.

On a more abstract level, we see this work as a contribution to bridging the gap between
second-order quantifier elimination (SOQE) and solving formula equations (FEQ) which will be of interest to the
respective communities investigating these problems.
SOQE has been studied mostly in the context of modal logic, description logics, knowledge representation and answer set programming.
On the other hand, FEQ is studied (usually in different formalisms and under different names)
in the verification and automated (inductive) theorem proving communities.
We believe that the complex of problems consisting of SOQE, WSOQE, and FEQ provides a suitable
  {\em common logical foundation} for work that is done on all of these topics, as has also been observed in~\cite{Wernhard17Boolean}.
This perspective not only allows to relate these, at first sight different, problems to one another
on a logical level, but also stimulates a cross-fertilization in terms of practical solution
ideas, techniques, and algorithms.

\subsubsection{Acknowledgements.}

We thank Luke Sanderson who dockerfied the original version of SCAN with help from the University of Manchester Open Source Software Club.

Also, we thank the anonymous reviewers for many helpful comments that led to an improvement of the presentation of this paper.

This research was funded in part by the Austrian Science Fund (FWF) 10.55776/P35787.

\subsubsection{Disclosure of Interests.}
The authors have no competing interests to declare that are relevant to the content of this article.

\bibliographystyle{splncs04}
\bibliography{./main.bib}

\clearpage
\appendix

\section{Relationship between WSOQE and FEQ}
We give a proof of a relationship between WSOQE and FEQ which justifies using WSOQE algorithms for solving FEQ by appending a validity check using, e.g., a first-order theorem prover.
\begin{restatable}{proposition}{relationFeqWsoqe}\label{prop.FEQWSOQE}
  Consider a formula of the form $\exists \seq{X}\, \varphi$ where $\varphi$ is first-order.
  Then $\exists \seq{X}\, \varphi$ has an FEQ-solution iff it is valid and has a WSOQE-witness.
\end{restatable}%
\begin{proof}
  For the left-to-right direction, let $\seq{\witness}$ be a tuple of first-order predicates
  such that\ $\models \varphi[\seq{X} \leftarrow \seq{\witness}]$.
  Then $\models \Exists{\seq{X}} \varphi$ and thus $\Exists{\seq{X}} \varphi \iff \varphi[\seq{X} \leftarrow \seq{\witness}]$.

  For the right-to-left direction, there is a tuple $\seq{\witness}$ of first-order predicates
  such that\ $\Exists{\seq{X}} \varphi \iff \varphi[\seq{X} \leftarrow \seq{\witness}]$.
  Since also $\models \Exists{\seq{X}} \varphi$, we obtain $\models \varphi[\seq{X} \leftarrow \seq{\witness}]$.
\end{proof}

\section{Proof of Main Result}
\label{sec.proofs}

In this section, we prove the main result of this paper:
\mainresult*

The main lemma to prove \Cref{witnesses-from-eliminating-scan-derivations} is the witness preservation lemma which shows that $T_S$ preserves witnesses backwards across a $\mathcal{C}$-derivation step:
\witnessPreservationLemma*
We first show \Cref{witnesses-from-eliminating-scan-derivations} from \Cref{witness-preservation-lemma} and give the proof of \Cref{witness-preservation-lemma} in the rest of this section.

\begin{proof}[Proof of \Cref{witnesses-from-eliminating-scan-derivations}]
  We show the statement by an inductive argument on the length of $D$.
  If $D$ is the empty sequence, then the conclusion of $D$ is $N$ and since $D$ is $\seq{X}$-eliminating we have that $N$ contains none of the variables in $\seq{X}$.
  Therefore $\seq{\mathrm{wit}(D)} = (W_1, \dots, W_d)$ satisfies $\Exists{\seq{X}} N \iff N[\seq{X} \leftarrow \seq{\mathrm{wit}(D)}]$.
  If $D = (D_i)_{1 \leq i \leq m+1}$, then $D' := (D_i)_{2 \leq i \leq m+1}$ is an $\seq{X}$-eliminating $\mathcal{C}$-derivation from $N(D)_1$ of length $m$.
  Thus by the induction hypothesis we have that $\seq{\mathrm{wit}(D')}$ is a witness for $\Exists{\seq{X}} N(D)_1$.
  Note that $\seq{\mathrm{wit}(D)} = T_{D_1}(\seq{\mathrm{wit}_1(D)}) = T_{D_1}(\seq{\mathrm{wit}(D')})$.
  Thus by \Cref{witness-preservation-lemma} we get that $\seq{\mathrm{wit}(D)}$ is a witness for $\Exists{\seq{X}} N$. 
\end{proof}

For formulas $\varphi$, $\psi$ we write $\varphi \imp \psi$ for validity of logical implication, i.e., for $\models \varphi \limp \psi$.
$\varphi \imp \psi$ implies $\varphi \models \psi$
and if $\varphi$ and $\psi$ are closed formulas then $\varphi \imp \psi$ is equivalent to $\varphi \models \psi$.
However, if $\varphi, \psi$ contain free variables then $\varphi \models \psi$ need not imply $\varphi \imp \psi$.
As the formulas and clause sets we deal with often contain free predicate variables and we often want to express this stronger notion of entailment, this is a convenient shorthand.
Note that $\varphi \iff \psi$ is equivalent to $\varphi \imp \psi$ and $\psi \imp \varphi$.

As a first step to showing the witness preservation lemma we prove the soundness of $\mathcal{C}$.

\begin{lemma}[Soundness of $\mathcal{C}$]
  \label{soundness-of-calculus}
  If $N/N'$ is a $\mathcal{C}$-derivation step, then $N \imp N'$.
\end{lemma}
\begin{proof}
  We first show soundness of all SCAN inference rules.
  Let $\mathcal{M}$ be a model.

  For constraint resolution assume $\mathcal{M} \models C \lor L(\seq{t})$ and $\mathcal{M} \models C' \lor L(\seq{s})^\perp$ and $\mathcal{M} \not\models \seq{t} \noeq \seq{s} \lor C \lor C'$.
  Then $\mathcal{M} \models \seq{t} \oeq \seq{s}$ and $\mathcal{M} \models C' \lor L(\seq{t})^\perp$.
  Thus, by soundness of propositional resolution we get $\mathcal{M} \models C \lor C'$ contradicting $\mathcal{M} \not\models \seq{t} \noeq \seq{s} \lor C \lor C'$.

  For factoring assume $\mathcal{M} \models C \lor L(\seq{t}) \lor L(\seq{s})$ and $\mathcal{M} \not\models \seq{t} \noeq \seq{s} \lor C \lor L(\seq{t})$.
  Then $\mathcal{M} \models \seq{t} \oeq \seq{s}$ and thus $\mathcal{M} \models C \lor L(\seq{t}) \lor L(\seq{s})$ contradicting $\mathcal{M} \not\models \seq{t} \noeq \seq{s} \lor C \lor L(\seq{t})$.

  For constraint elimination assume $\mathcal{M} \models \seq{t} \noeq \seq{s} \lor C$ and $\sigma$ is a most general unifier of $\seq{t}$ and $\seq{s}$ and $\mathcal{M} \not\models C\sigma$
  If $\mathcal{M} \models \seq{t} \noeq \seq{s}$, then $\seq{t}$ and $\seq{s}$ are not unifiable, contradicting the existence of the unifier $\sigma$.
  If $\mathcal{M} \models C$, then $\mathcal{M} \models C\theta$ for all substitutions $\theta$, contradicting $\mathcal{M} \not\models C\sigma$.
  
  Now if $S \in \mset{\mathrm{Res}, \mathrm{Fac}, \mathrm{ConstrElim}}$ is a $\mathcal{C}$-derivation step we get $N \imp N'$ by the just proven soundness of the inference steps.

  If $S = \mathrm{RedElim}$ we get $N \iff N'$ by redundancy being equivalence-preserving and thus $N \imp N'$.
  
  For the other derivation steps $\mathrm{ExtPurDel}_X^{+}, \mathrm{ExtPurDel}_X^{-}$ and $\mathrm{PurDel}_P$ we have $N' \subseteq N$ so soundness follows immediately.
\end{proof}

Instead of showing the witness preservation lemma directly we will show a slightly more general version of it by introducing a different version of $T_S$ which does not take a witness as input, but instead uses a second-order variable as a placeholder for it.
This will also be useful for correctness proof of SCAN (\Cref{correctness-of-scan}) where we not only substitute first-order predicates for the second-order variables, but also actual predicates in the context of a model.
\begin{restatable}[Witness transformation]{definition}{witnessTransformationPrimeDefinition}
  Let $\seq{X} = (X_1, \dots, X_d)$ be a tuple of predicate variables.
  We set
  $$
  \begin{array}{rll}
    T_{S}' &:= &\seq{X} \text{ for $S \in \mset{\mathrm{Res}, \mathrm{Fac}, \mathrm{ConstrElim}, \mathrm{RedElim}}$}, \\
    T_{\mathrm{ExtPurDel}^{+}_{X_i}}' &:= &(X_1, \dots, X_{i-1}, \lambda \seq{u_i}. \top, X_{i+1}, \dots, X_d)\\
    T_{\mathrm{ExtPurDel}^{-}_{X_i}}' &:= &(X_1, \dots, X_{i-1}, \lambda \seq{u_i}. \bot, X_{i+1}, \dots, X_d)\\
    T_{\mathrm{PurDel_P}}' &:= &(X_1, \dots, X_{i-1}, \resclosurepredicate{P}, X_{i+1}, \dots, X_d) \\
    & &\text{if $P$ is $X_i$-pointed}.
  \end{array}
  $$
\end{restatable}%
With this notation we have $T_S'[\seq{X} \leftarrow \seq{\witness}] = T_S(\seq{\witness})$.
The core property we want to show about $T_S'$ is the following:
\begin{restatable}[Transformation lemma]{lemma}{transformationLemma}
  \label{transformation-lemma}
  If $S$ is a $\mathcal{C}$-derivation step from $N$ to $N'$, then $N' \imp N[\seq{X} \leftarrow T_S']$.
\end{restatable}

Before proving the transformation lemma we show how the witness preservation lemma (\Cref{witness-preservation-lemma}) follows from the transformation lemma:
\begin{proof}
  Let $\seq{\witness}$ be a witness for $\Exists{\seq{X}} N'$.
  We show both directions of the WSOQE-condition \eqref{wsoqe-condition} $\Exists{\seq{X}} N \iff N[\seq{X} \leftarrow T_S(\seq{\witness})]$.

  The direction $N[\seq{X} \leftarrow T_S(\seq{\witness})] \imp \Exists{\seq{X}} N$ follows as $T_S(\seq{\witness})$ is an instance of the existential quantifiers.
  
  It remains to show $\Exists{\seq{X}} N \imp N[\seq{X} \leftarrow T_S(\seq{\witness})]$.
  By soundness of $\mathcal{C}$ (\Cref{soundness-of-calculus}) we have $N \imp N'$ thus it suffices to show $\Exists{\seq{X}} N' \imp N[\seq{X} \leftarrow T_S(\seq{\witness})]$.
  Since $\seq{\witness}$ is a witness for $\Exists{\seq{X}} N'$ it even suffices to show $N'[\seq{X} \leftarrow \seq{\witness}] \imp N[\seq{X} \leftarrow T_S(\seq{\witness})]$.
  By the transformation lemma we have $N'[\seq{X} \leftarrow \seq{\witness}] \imp N[\seq{X} \leftarrow T_S'][\seq{X} \leftarrow \seq{\witness}]$. In other words, $N'[\seq{X} \leftarrow \seq{\witness}] \imp N[\seq{X} \leftarrow T_S(\seq{\witness})]$.
\end{proof}

It remains to show the transformation lemma (\Cref{transformation-lemma}).
We first show it for the $\mathcal{C}$-derivation steps other than purified clause deletion in which case we even have equivalence between the two clause sets:
\begin{lemma}
  \label{transformation-lemma.except-purdel}
  If $S$ is a $\mathcal{C}$-derivation step from $N$ to $N'$ which is not a purified clause deletion step, then $N[\seq{X} \leftarrow T_S'] \iff N'$.
\end{lemma}
\begin{proof}
  We prove this by a case distinction on $S$.

  For $S \in \mset{\mathrm{Res}, \mathrm{Fac}, \mathrm{ConstrElim}, \mathrm{RedElim}}$ we have $T_S' = \seq{X}$, thus we need to show $N \iff N'$.
  We get $N \imp N'$ by soundness of $\mathcal{C}$ (\Cref{soundness-of-calculus}).
  If $S$ is one of $\mathrm{Res}, \mathrm{Fac}$ or $\mathrm{ConstrElim}$ we get $N' \imp N$ by $N \subseteq N'$.
  For $S = \mathrm{RedElim}$ we have $N \iff N'$ since the notion of redundancy is equivalence-preserving.

  Now consider $S = \mathrm{ExtPurDel}_{X_i}^{+}$.
  Then $N' = N \setminus \mset{C \in N \suchthat \text{$C$ contains $X_i$}}$ and every clause in $N$ containing $X_i$ contains $X_i$ with positive polarity.
  Furthermore $T_S' = (X_1, \dots, X_{i-1}, \lambda \seq{u}. \top, X_{i+1}, \dots, X_d)$ and therefore we get that $N[\seq{X} \leftarrow T_S'] = N[X_i \leftarrow \lambda \seq{u}. \top]$.
  We first show $N[X_i \leftarrow \lambda \seq{u}. \top] \imp N'$.
  Since $N' \subseteq N$ we have $N \imp N'$.
  Applying the substitution $[X_i \leftarrow \lambda \seq{u}. \top]$ on both sides yields $N[X_i \leftarrow \lambda \seq{u}. \top] \imp N'[X_i \leftarrow \lambda \seq{u}. \top]$.
  Since $N'$ does not contain $X_i$ we get $N[X_i \leftarrow \lambda \seq{u}. \top] \imp N'$.
  It remains to show $N' \imp N[X_i \leftarrow \lambda \seq{u}. \top]$ which we do by showing $N' \imp C[X_i \leftarrow \lambda \seq{u}. \top]$ for all $C \in N$.
  If $C$ contains $X_i$, it contains $X_i$ positively and we get $\models C[X_i \leftarrow \lambda \seq{u}. \top]$.
  If $C$ does not contain $X_i$, then $C \in N'$ and $C[X_i \leftarrow \lambda \seq{u}. \top] = C$ and therefore $N' \imp C[X_i \leftarrow \lambda \seq{u}. \top]$.
  
  The proof for the case $S = \mathrm{ExtPurDel}_{X_i}^{-}$ is analogous to the proof for the positive case.
\end{proof}

Now it remains to show the transformation lemma for the case of purified clause deletion steps.
In order to do that we first study the structure of $\resclosure{P}{C}$ in relation to $\resclosurepremise{P}{\seq{c}}$, captured in the following lemma:
\begin{lemma}
  \label{resclosure-structure}
  Let $P = \underline{L(\seq{t})} \lor C_P$ be an $X$-pointed clause and let $C$ be a clause. Then 
  
  \begin{enumerate}[(i)]
    \item \label{decomposition} $C = C' \lor \Lor_{i=1}^n L(\seq{t_i})^{\perp}$
  for some terms $\seq{t_1}, \dots, \seq{t_n}$ and a clause $C'$ where $X$ does not occur with polarity opposite to $L$.
    \item \label{resclosure-closed-under-substitution-of-clauses-from-resclosurepremise}$\resclosure{P}{C}$ is closed under substitution of $L^\perp$ for clauses from $\resclosurepremise{P}{\seq{c}}$, i.e., if $C_0 \lor L(\seq{s})^\perp \in \resclosure{P}{C}$ and $R(\seq{c}) \in \resclosurepremise{P}{\seq{c}}$, then $C_0 \lor R(\seq{s}) \in \resclosure{P}{C}$.
    \item \label{resclosure-structure.main} 
    $$
      \resclosure{P}{C} = \left\{ C' \lor \Lor_{i=1}^n R_i(\seq{t_i}) \suchthat R_1(\seq{c}), \dots, R_n(\seq{c}) \in \resclosurepremise{P}{\seq{c}} \right\}.
    $$
  \end{enumerate}
\end{lemma}
\begin{proof}
  \begin{enumerate}[(i)]
    \item Follows by splitting $C$ into the literals that have polarity opposite to $L$ and the remaining literals.
    
    \item We show this statement by induction on the structure of $\resclosurepremise{P}{C}$:
    The statement follows immediately from the premise for $R(\seq{c}) = L(\seq{c})^\perp$.
    
    It remains to show that when $R(\seq{c})$ is a resolvent between $P$ and a clause $R'(\seq{c}) \in \resclosurepremise{P}{\seq{c}}$ such that $C_0 \lor R'(\seq{s}) \in \resclosure{P}{C}$, then we get that $C_0 \lor R(\seq{s}) \in \resclosure{P}{C}$.
    Since $R$ is a resolvent between $P$ and $R'(\seq{c})$ we have that
    \begin{prooftree}
      \AxiomC{$P$}
      \AxiomC{$R'(\seq{c})$}
      \RightLabel{$\mathrm{Res}$}
      \BinaryInfC{$R(\seq{c})$}
    \end{prooftree}
    is a constraint resolution inference in $\mathcal{C}^\ast$ with the two resolution premises $P$ and $P_R(\seq{c}) = \underline{L(\seq{r(\seq{c})})^\perp} \lor C_R(\seq{c})$.
    Then by resolving on the same literal in $R'$ we get that
    \begin{prooftree}
      \AxiomC{$P$}
      \AxiomC{$C_0 \lor R'(\seq{s})$}
      \RightLabel{$\mathrm{Res}$}
      \BinaryInfC{$C_0 \lor R(\seq{s})$}
    \end{prooftree}
    is a constraint resolution inference in $\mathcal{C}^\ast$ with resolution premises $P$ and $C_0 \lor P_R(\seq{s})$.
    Since also $C_0 \lor R'(\seq{s}) \in \resclosure{P}{C}$ we get
    $C_0 \lor R(\seq{s}) \in \resclosure{P}{C}$.

    \item $\subseteq$: We show that the set on the right hand side contains $C$ and is closed under resolution with $P$.
  Since $\resclosure{P}{C}$ is the smallest such set the statement then follows.
  Recall that $\resclosurepremise{P}{\seq{c}} = \resclosure{P}{L(\seq{c})^\perp}$.
  Therefore we have $L(\seq{c})^\perp \in \resclosurepremise{P}{\seq{c}}$ which shows that $C$ is in the right hand set.

  Now let $R_1(\seq{c}), \dots, R_n(\seq{c}) \in \resclosurepremise{P}{\seq{c}}$ and let $R$ be a resolvent between $P$ and $C'' := C' \lor \Lor_{i=1}^n R_i(\seq{t_i})$.
  We need to show that $R$ is in the right hand set.
  Since $C'$ does not contain $X$ with polarity opposite to $L$ we have that the resolved upon literal in $C''$ must be in one of the $R_i$, i.e., $R_j = L(\seq{s})^\perp \lor R_j'$ for some $1 \leq j \leq n$.
  Now let $R'(\seq{c}) \in \resclosurepremise{P}{\seq{c}}$ be the resolvent between $P$ and $R_j(\seq{c}) \in \resclosurepremise{P}{\seq{c}}$.
  Then 
  $$R = C' \lor \Lor_{\stackrel{i=1}{i \neq j}}^n R_i(\seq{t_i}) \lor R'(\seq{t_j})$$ which shows that $R$ is contained in the right hand set.

  $\supseteq$:
  We now show the following statement by induction on $0 \leq j \leq n$:
  For all $R_1(\seq{c}), \dots, R_j(\seq{c}) \in \resclosurepremise{P}{\seq{c}}$ we have
  $$
  C' \lor \Lor_{i=1}^j R_j(\seq{t_i}) \lor \Lor_{i=j+1}^n L(\seq{t_i})^\perp \in \resclosure{P}{C}.
  $$
  The desired statement then follows for $j=n$.

  For $j = 0$ the statement becomes $C' \lor \Lor_{i=1}^n L(\seq{t_i})^\perp \in \resclosure{P}{C}$, in other words $C \in \resclosure{P}{C}$ by (\ref{decomposition}), which holds by definition of $\resclosure{P}{C}$.

  For the induction step we can assume 
  $$
  C' \lor \Lor_{i=1}^j R_j(\seq{t_i}) \lor \Lor_{i=j+1}^n L(\seq{t_i})^\perp \in \resclosure{P}{C}.
  $$
  Then by applying (\ref{resclosure-closed-under-substitution-of-clauses-from-resclosurepremise}) for the literal $L(\seq{t_{j+1}})^\perp$ we get  
  $$
  C' \lor \Lor_{i=1}^{j+1} R_j(\seq{t_i}) \lor \Lor_{i=j+2}^n L(\seq{t_i})^\perp \in \resclosure{P}{C}
  $$
  which is what we needed to show.
  \end{enumerate}
\end{proof}

Note that $\resclosurepremise{P}{\seq{c}}$ and $\resclosurepredicate{P}$ satisfy the following duality properties which will help us avoid repetition in the subsequent proofs.
In the following, if $\alpha = \lambda \seq{u}. \varphi$ is a predicate expression we write $\neg \alpha$ for $\lambda \seq{u}. \neg \varphi$.
\begin{lemma}
  \label{duality-lemma}
  Let $P = \underline{L} \lor C_P$ be an $X$-pointed clause.
  For a clause $C$ denote by $C^\perp$ the clause which results from $C$ by dualizing all $X$-literals.
  Furthermore let $P^\perp := \underline{L^\perp} \lor C_P^\perp$.
  Then
  \begin{enumerate}[(i)]
    \item \label{resclosure-duality} For all clauses $C$ we have $\resclosure{P^\perp}{C^\perp} = \mset{R^\perp \suchthat R \in \resclosure{P}{C}}$.
    \item \label{resclosurepremise-duality} $\resclosurepremise{P^\perp}{\seq{c}} = \mset{R^\perp \suchthat R \in \resclosurepremise{P}{\seq{c}}}$ and
    \item \label{resclosurepredicate-duality} $\resclosurepredicate{P} \iff \neg \resclosurepredicate{P^\perp}[X \leftarrow \lambda \seq{u}. \neg X(\seq{u})]$.
  \end{enumerate}
\end{lemma}
\begin{proof}
  \begin{enumerate}[(i)]
    \item $\subseteq$: To show that $\resclosurepremise{P^\perp}{C^\perp}$ is contained in the right-hand set we show that the right-hand set contains $C^\perp$ and is closed under resolution with $P^\perp$.
          Since $\resclosurepremise{P^\perp}{C^\perp}$ is the smallest such set we get the desired statement.

          Since $C \in \resclosurepremise{P}{C}$ we get $C^\perp$ is in the right-hand set.
          
          Now let $R \in \resclosurepremise{P}{C}$ such that
          $$
          \AxiomC{$R^\perp$}
          \AxiomC{$P^\perp$}
          \RightLabel{$\mathrm{Res}$}
          \BinaryInfC{$R_P$}
          \DisplayProof
          $$
          is a constraint resolution inference in $\mathcal{C}^\ast$ on the designated literal of $P^\perp$.
          Then note that 
          $$
          \AxiomC{$R$}
          \AxiomC{$P$}
          \RightLabel{$\mathrm{Res}$}
          \BinaryInfC{$R_P^\perp$}
          \DisplayProof
          $$
          is a constraint resolution inference in $\mathcal{C}^\ast$ on the designated literal of $P$.
          Therefore $R_P^\perp \in \resclosure{P}{C}$.
          Since ${R_P^\perp}^\perp = R_P$ we get that $R_P$ is in the right-hand set.

          $\supseteq$: To show that the right-hand set is contained in  $\resclosurepremise{P^\perp}{C^\perp}$ we show that $R^\perp \in \resclosurepremise{P^\perp}{C^\perp}$ for all $R \in \resclosurepremise{P}{C}$ by induction on the definition of $\resclosurepremise{P}{C}$.

          If $R = C$ we get $R^\perp \in \resclosurepremise{P^\perp}{C^\perp}$ by definition.

          Now let 
          $$
          \AxiomC{$R_0$}
          \AxiomC{$P$}
          \RightLabel{$\mathrm{Res}$}
          \BinaryInfC{$R$}
          \DisplayProof
          $$
          be a constraint resolution inference in $\mathcal{C}^\ast$ on the designated literal of $P$ and $R_0 \in \resclosurepremise{P}{C}$.
          Then by induction hypothesis we get $R_0^\perp \in \resclosurepremise{P^\perp}{C^\perp}$.
          Now 
          $$
          \AxiomC{$R_0^\perp$}
          \AxiomC{$P^\perp$}
          \RightLabel{$\mathrm{Res}$}
          \BinaryInfC{$R^\perp$}
          \DisplayProof
          $$
          is a constraint resolution inference in $\mathcal{C}^\ast$ on the designated literal of $P^\perp$.
          Since $\resclosurepremise{P^\perp}{C^\perp}$ is closed under resolvents with $P$ we can conclude that $R^\perp \in \resclosurepremise{P^\perp}{C^\perp}$.
    \item Recall that $\resclosurepremise{P}{\seq{c}} = \resclosure{P}{L(\seq{c})^\perp}$ where $L(\seq{t})$ is the designated literal of $P$. 
          We now apply (\ref{resclosure-duality}) for $C = L(\seq{c})$ to get the desired statement.
    \item We denote by $\eta$ the substitution which sends $X$ to $\lambda \seq{u}. \neg X(\seq{u})$.
    First we show the statement for the case where $L$ is a positive literal, meaning $L^\perp$ is a negative literal.
    Then by definition of $\resclosurepredicate{P^\perp}$ we get
    \begin{align*}
      \neg \resclosurepredicate{P^\perp}\eta
      \iff& \lambda \seq{u}.\neg \Land_{R(\seq{c},\seq{v}) \in \resclosurepremise{P^\perp}{\seq{c}}} \Forall{\seq{v}} R(\seq{u}, \seq{v})\eta. \\
      \intertext{Using standard equivalences in predicate logic results in}
      \neg \resclosurepredicate{P^\perp}\eta \iff& \lambda \seq{u}. \Lor_{R(\seq{c}, \seq{v}) \in \resclosurepremise{P^\perp}{\seq{c}}} \Exists{\seq{v}} \neg (R(\seq{u}, \seq{v})\eta). \\
      \intertext{By (\ref{resclosurepremise-duality}) we have}
      \neg \resclosurepredicate{P^\perp}\eta \iff& \lambda \seq{u}. \Lor_{R(\seq{c}, \seq{v}) \in \resclosurepremise{P}{\seq{c}}} \Exists{\seq{v}} \neg (R(\seq{u}, \seq{v})^\perp\eta). \\
      \intertext{Elimination of double negation ($R^\perp\eta \iff R$) results in}
      \neg \resclosurepredicate{P^\perp}\eta \iff& \lambda \seq{u}. \Lor_{R(\seq{c}, \seq{v}) \in \resclosurepremise{P}{\seq{c}}} \Exists{\seq{v}} \neg R(\seq{u}, \seq{v}).
    \end{align*}
    Finally, by definition of $\resclosurepredicate{P}$ we get $\neg \resclosurepredicate{P^\perp}\eta \iff \resclosurepredicate{P}$.

    Now let $L$ be a negative $X$-literal.
    Then the designated literal of $P^\perp$ is positive so we can apply the already proved statement to it and get
    \begin{align*}
      &\resclosurepredicate{P^\perp} \iff \neg \resclosurepredicate{{P^\perp}^\perp}\eta. \\
      \intertext{Note that ${P^\perp}^\perp = P$, thus we get} 
      &\resclosurepredicate{P^\perp} \iff \neg \resclosurepredicate{P}\eta. \\
      \intertext{Applying negation on both sides and eliminating double negation yields}
      \neg&\resclosurepredicate{P^\perp} \iff \resclosurepredicate{P}\eta. \\
      \intertext{Applying $\eta$ on both sides and eliminating double negation results in}
      \neg&\resclosurepredicate{P^\perp}\eta \iff \resclosurepredicate{P}.
    \end{align*}
    which is what we wanted to show.
  \end{enumerate}
  
\end{proof}

With this we can now prove a crucial property about $\resclosurepredicate{P}$, already mentioned in the main text.
In the proof we use the following notation: 
For an expression $F = \lambda \seq{u}. E(\seq{u})$ and a tuple of terms $\seq{t}$ of the same size as $\seq{u}$ we denote by $F[\seq{t}]$ the expression $E(\seq{t})$, i.e., the substitution of the terms $\seq{t}$ for the $\lambda$-abstracted variables $\seq{u}$.

\resclosurepredicateLemma*
\begin{proof}
  \label{resclosurepredicate-lemma.proof}
  First, we show how to get the first statement from the second statement and then prove the second statement.

  Also, we first consider the case where $P = \underline{\neg X(\seq{t})} \lor C'$, i.e., the designated literal of $P$ is negative.
  We use the second statement for the clause $C=\seq{c} \noeq \seq{t} \lor C'$ where $\seq{c}$ is a tuple of fresh constants.
  Then we get 
  $$\resclosure{P}{\seq{c} \noeq \seq{t} \lor C'}~\imp~\seq{c} \noeq \seq{t} \lor C'[X~\leftarrow~\resclosurepredicate{P}].$$
  Resolving $P$ with $X(\seq{c}) \in \resclosurepremise{P}{\seq{c}}$ results in $\seq{c} \noeq \seq{t} \lor C' \in \resclosurepremise{P}{\seq{c}}$ which implies $\resclosure{P}{\seq{c} \noeq \seq{t} \lor C'}~\subseteq~\resclosurepremise{P}{\seq{c}}$.
  Thus 
  $$\resclosurepremise{P}{\seq{c}} \imp \seq{c} \noeq \seq{t} \lor C'[X \leftarrow \resclosurepredicate{P}]$$
  which by $\resclosurepremise{P}{\seq{c}} \iff \resclosurepredicate{P}[\seq{c}]$ is equivalent to 
  $$\models \neg \resclosurepredicate{P}[\seq{c}] \lor \seq{c} \noeq \seq{t} \lor C'[X \leftarrow \resclosurepredicate{P}].$$
  Eliminating the constraint $\seq{c} \noeq \seq{t}$ gives 
  $$\models \neg \resclosurepredicate{P}[\seq{t}] \lor C'[X \leftarrow \resclosurepredicate{P}]$$ which precisely means $\models P[X \leftarrow \resclosurepredicate{P}]$.
  
  Now we show the case where the designated literal of $P$ is positive, i.e., $P = \underline{X(\seq{t})} \lor C'$.
  Then by the negative case we get that $\models P^{\perp}[X \leftarrow \resclosurepredicate{P^{\perp}}]$.
  Since this holds for arbitrary $X$ we can substitute $X$ by $\neg X$ and get $\models P^{\perp}[X \leftarrow \resclosurepredicate{P^{\perp}}][X \leftarrow \lambda \seq{u}. \neg X(\seq{u})]$, i.e., $\models P^{\perp}[X \leftarrow \resclosurepredicate{P^{\perp}}[X \leftarrow \lambda \seq{u}. \neg X(\seq{u})]]$.
  By \Cref{duality-lemma} (\ref{resclosurepredicate-duality}) this means $\models P^{\perp}[X \leftarrow \lambda \seq{u}. \neg \resclosurepredicate{P}[\seq{u}]]$.
  Eliminating the double negation of dualization ${}^\perp$ on $X$-literals and the negation sign of the substitution gives us the desired statement $\models P[X \leftarrow \resclosurepredicate{P}]$.

  Now we show the second statement.
  We first give a proof for the case where $P = \underline{\neg X(\seq{t})} \lor C'$, i.e., the designated literal of $P$ is negative.
  By \Cref{resclosure-structure} (\ref{decomposition}) we have $C = C'' \lor \Lor_{i=1}^n X(\seq{t_i})$ for some clause $C''$ where $X$ does not occur positively. 
  Now assume there is a model $\mathcal{M}$ and an assignment $\theta$ of the free predicate variables in $\resclosure{P}{C}$ such that $\mathcal{M}, \theta \models \resclosure{P}{C}$ and furthermore $\mathcal{M}, \theta \not\models C[X \leftarrow \resclosurepredicate{P}]$.
  This means
  \begin{enumerate}[(a)]
    \item \label{a} $\mathcal{M}, \theta \models \resclosure{P}{C}$,
    \item \label{b} $\mathcal{M}, \theta \not\models C''[X \leftarrow \resclosurepredicate{P}]$ and
    \item \label{c} $\mathcal{M}, \theta \not\models \Lor_{i=1}^n\resclosurepredicate{P}[\seq{t_i}]$.
  \end{enumerate}

  By \Cref{resclosure-structure} all clauses of $\resclosure{P}{C}$ have the form 
  $C'' \lor \Lor_{i=1}^n C_i(\seq{t_i})$ where $C_1(\seq{c}), \dots, C_n(\seq{c}) \in \resclosurepremise{P}{\seq{c}}$.
  Thus from (\ref{a}) we get 
  \begin{equation}
    \label{one}
    \mathcal{M}, \theta \models C'' \lor \Lor_{i=1}^n C_i(\seq{t_i}) \text{ for all $C_1(\seq{c}), \dots, C_n(\seq{c}) \in \resclosurepremise{P}{\seq{c}}$.}    
  \end{equation}

  Note that since $X(\seq{c}) \in \resclosurepremise{P}{\seq{c}}$ we get $\models \Forall{\seq{u}} (\resclosurepredicate{P}[\seq{u}] \limp X(\seq{u}))$. 
  Since $X$ does not occur positively in $C''$ we get $\models C'' \limp C''[X \leftarrow \resclosurepredicate{P}]$ by a monotonicity argument.
  Thus from (\ref{b}) we get 
  \begin{equation}
    \label{two}
    \mathcal{M}, \theta \not\models C''.
  \end{equation}

  Remember that $\resclosurepredicate{P} = \lambda \seq{u}. \Land_{R(\seq{c}, \seq{v}) \in \resclosurepremise{P}{\seq{c}}} \Forall{\seq{v}} R(\seq{u}, \seq{v})$.
  Thus by (\ref{c}) we get that for all $i \in \mset{1, \dots n}$ there is a $C_i(\seq{c}) \in \resclosurepremise{P}{\seq{c}}$ such that $\mathcal{M} \not\models C_i(\seq{t_i})$ which together with \eqref{two} gives a contradiction to \eqref{one}. 
  
  We now prove the second statement for the case where the designated literal of $P$ is positive.
  Then the designated literal of $P^\perp$ is negative.
  We now use the second statement on $P^\perp$ and $C^\perp$ as in \Cref{duality-lemma} to get
  \begin{equation}
    \label{eq.resclosure-implication}
  \models \resclosure{P^\perp}{C^\perp} \limp C^\perp[X \leftarrow \resclosurepredicate{P^\perp}].
  \end{equation}
  Now let $\eta$ be the substitution which sends $X$ to $\lambda \seq{u}. \neg X(\seq{u})$.  
  By \Cref{duality-lemma} (\ref{resclosure-duality}) we get $\resclosure{P^\perp}{C^\perp} = \{{C'}^\perp \suchthat C' \in \resclosure{P}{C}\}$.
  Applying $\eta$ to all clauses in $\resclosure{P^\perp}{C^\perp}$ results in $\{{C'}^\perp\eta \suchthat C' \in \resclosure{P}{C}\}$ which by $C'^\perp \eta \iff C'$ is logically equivalent to $\resclosure{P}{C}$.
  Note that $C^\perp[X \leftarrow \resclosurepredicate{P^\perp}] \iff C[X \leftarrow \neg \resclosurepredicate{P^\perp}]$.
  Now applying $\eta$ to $C[X \leftarrow \neg \resclosurepredicate{P^\perp}]$ results in 
  $$C[X \leftarrow \neg \resclosurepredicate{P^\perp}]\eta = C[X \leftarrow \neg \resclosurepredicate{P^\perp}\eta]$$
  which by \Cref{duality-lemma} (\ref{resclosurepredicate-duality}) is equivalent to $C[X \leftarrow \resclosurepredicate{P}]$. 
  Thus applying $\eta$ to \eqref{eq.resclosure-implication} results in the desired statement.
\end{proof}

We can now prove the full transformation lemma which completes the proof of our main result.
\transformationLemma*
\begin{proof}
  If $S$ is not a purified clause the result follows from \Cref{transformation-lemma.except-purdel}.
  
  It remains to show the case where $S = \mathrm{PurDel}_P$ for some pointed clause $P$.
  Let $\seq{X} = (X_1, \dots, X_d)$ and let $N/N'$ be a $\mathrm{PurDel}_P$ step.
  Then we have $N = N' \union \mset{P}$ and $P$ is purified in $N'$ and $X_i$-pointed for some $i \in \mset{1, \dots, d}$.
  Therefore we have $T_{S}' = (X_1, \dots, X_{i-1}, \resclosurepredicate{P}, X_{i+1}, \dots, X_d)$, so it suffices to show the implication $N' \imp N[X_i \leftarrow \resclosurepredicate{P}]$.
  Using the first statement of \Cref{resclosurepredicate-lemma} we get $\models P[X_i \leftarrow \resclosurepredicate{P}]$.
  Thus it remains to show the statement $N' \imp N'[X_i \leftarrow \resclosurepredicate{P}]$.
  Since $P$ is purified in $N'$ we have that $\resclosure{P}{N'}$ is redundant in $N'$ and thus $N' \imp \resclosure{P}{N'}$ and in particular $N' \imp \resclosure{P}{C}$ for all $C \in N'$.
  Now by second statement of \Cref{resclosurepredicate-lemma} we get the implication $N' \imp C[X_i \leftarrow \resclosurepredicate{P}]$ which finishes the proof. 
\end{proof}

\section{Correctness of SCAN}
\label{sec.correctness-of-scan}

Using the transformation lemma from the previous section we can give a new correctness proof for SCAN:
\calculusIsSoundAndExistentialEquivalencePreserving*
\begin{proof}
  By soundness of $\mathcal{C}$ (\Cref{soundness-of-calculus}) we get $N \imp N'$ and thus $\Exists{\seq{X}} N \imp \Exists{\seq{X}} N'$.
  Therefore it only remains to show $\Exists{\seq{X}} N' \imp \Exists{\seq{X}} N$.

  Let $\mathcal{M}$ be a model of $\Exists{\seq{X}} N'$ and let $\theta'$ be an assignment of the predicate variables $\seq{X}$ such that $\mathcal{M}, \theta' \models N'$.
  By the transformation lemma (\Cref{transformation-lemma}) we get $\mathcal{M}, \theta' \models N[\seq{X} \leftarrow T_S']$.
  Now, let $T_S' = (\lambda \seq{u_1}. \varphi_1'(\seq{u_1}), \dots, \lambda \seq{u_d}. \varphi_d'(\seq{u_d}))$ with formulas $\varphi_1', \dots, \varphi_d'$ where $\seq{u_i}$ is of size $n_i$, i.e., the $\varphi_i'$ define the predicate expressions in $T_S'$.
  Now set $Q_i := \mset{\seq{m} \in M^{n_i} \suchthat \mathcal{M}, \theta' \models \varphi_i'(\seq{m})}$ for $1\leq i \leq d$ and let $\theta$ be the assignment of the predicate variables $\seq{X}$ to the relations $\seq{Q}$.
  Since $\mathcal{M}, \theta' \models N[\seq{X} \leftarrow T_S']$ we get $\mathcal{M}, \theta \models N$, i.e., $\mathcal{M} \models \Exists{\seq{X}} N$.
\end{proof}

\section{First-order witnesses}

We show our main result for the finite case as a corollary to \Cref{witnesses-from-eliminating-scan-derivations}:
\finiteMainResult*
\begin{proof}
  Follows directly from \Cref{witnesses-from-eliminating-scan-derivations} since we replace predicate expressions by equivalent first-order predicates.
\end{proof}

We show that one-sidedness implies finiteness of $\resclosurepremise{P}{\seq{c}}$:
\begin{lemma}
  \label{one-sided-implies-ResU-finite}
  Let $P$ be a one-sided pointed clause.
  Then $\resclosurepremise{P}{\seq{c}}$ is finite.
\end{lemma}
\begin{proof}
  Let $P = \underline{L(\seq{t})} \lor C$.
  Since $P$ is one-sided we have that $C$ does not contain $L^\perp$-literals.
  Therefore
  $\resclosurepremise{P}{\seq{c}} = \mset{L(\seq{c})^\perp, \seq{c} \noeq \seq{t} \lor C}$ is finite.
\end{proof}
However, one-sidedness is not necessary for $\resclosurepremise{P}{\seq{c}}$ being logically equivalent to a finite clause set:
\begin{example}
  \label{mixed-resolution-candidate-with-finite-resolution-closure}
  $P = \underline{\neg X(u,v)} \lor X(v,u)$ is not one-sided, yet we have 
  \begin{align*}
    \resclosurepremise{P}{c,d} = \{ &X(c,d), \\
    &c \noeq u \lor d \noeq v \lor X(v,u), \\
    &u \noeq v' \lor v \noeq u' \lor c \noeq u' \lor d \noeq v' \lor X(v,u),\\
    &\dots\}
  \end{align*}
  which is logically equivalent to $\mset{X(c,d), X(d,c)}$.
\end{example}

For one-sided derivations we thus get:
\oneSidedMainResult*
\begin{proof}
  By \Cref{one-sided-implies-ResU-finite} we get that $F_P = \resclosurepredicate{P}$ is a first-order annotation for $D$.
  Also note that in this case $\seq{\mathrm{wit}(D,F)} = \seq{\mathrm{wit}(D)}$.
  Thus the result follows from \Cref{finite-witnesses-from-eliminating-scan-derivations}.
\end{proof}

\section{Size of witnesses}
\label{sec.size-of-witnesses}

To show the result about the size of witnesses of one-sided derivations we use a formal definition of size:
\begin{definition}
  Let $E$ be an expression or a tuple of expressions.
  We define the \emph{size of $E$} (denoted as $\abs{E}$) inductively as
  \begin{align*}
    \abs{E} &= 1, \text{ if $E$ is an individual constant, variable, $\bot$ or $\top$}  \\
    \abs{f(\seq{t})} &= 1 + \abs{\seq{t}} \text{ for function symbols $f$ and terms $\seq{t}$} \\
    \abs{R(\seq{t})} &= 1 +  \abs{\seq{t}}, \text{ for predicate symbols $R$ and terms $\seq{t}$} \\
    \abs{X(\seq{t})} &= 1 +  \abs{\seq{t}}, \text{ for predicate variables $X$ and terms $\seq{t}$} \\
    \abs{\neg \varphi} &= 1 + \abs{\varphi}, \text{ for a formula $\varphi$} \\
    \abs{\bigoplus_{i \in I} \varphi_i} &= \begin{cases}
      \sum_{i \in I} \abs{\varphi_i},& \text{ for $\oplus \in \mset{\land, \lor}$ and formulas $\varphi_i$, if $I$ is finite} \\
      \infty &\text{ for $\oplus \in \mset{\land, \lor}$ and formulas $\varphi_i$, if $I$ is infinite}
    \end{cases} \\
    \abs{\varphi \limp \psi} &= \abs{\varphi} + \abs{\psi}, \text{ for formulas $\varphi$ and $\psi$} \\
    \abs{\varphi \liff \psi} &= \abs{\varphi} + \abs{\psi}, \text{ for formulas $\varphi$ and $\psi$} \\
    \abs{\Forall{u} \varphi} &= 1 + \abs{\varphi}, \text{ for a formula $\varphi$} \\
    \abs{\Exists{u} \varphi} &= 1 + \abs{\varphi}, \text{ for a formula $\varphi$} \\
    \abs{\Forall{X} \varphi} &= 1 + \abs{\varphi}, \text{ for a formula $\varphi$} \\
    \abs{\Exists{X} \varphi} &= 1 + \abs{\varphi}, \text{ for a formula $\varphi$} \\
    \abs{\lambda u. E} &= \abs{E}, \text{ for an expression $E$} \\
    \abs{(E_1, \dots, E_d)} &= \sum_{i=1}^d \abs{E_i}, \text{ for expressions $E_1, \dots, E_d$.}
  \end{align*}
\end{definition}

For one-sided $\seq{X}$-eliminating $\mathcal{C}$-derivations we can give a tight upper bound on the size of the witnesses:
\begin{proposition}
  \label{size-of-witnesses-from-one-sided-eliminating-derivation}
  Let $D = (D_i)_{1 \leq i \leq m}$ be a one-sided $\seq{X}$-eliminating derivation from $N$.
  Let $p$ be the number of purified pointed clauses in $D$ and let $n$ be the maximum number of $\seq{X}$-literals in a purified pointed clause in $D$.
  Let $\mathrm{wit}(D) = (\witness_1(D), \dots, \witness_d(D))$.
  Then for all $1 \leq i \leq d$ we have
  \begin{enumerate}[(i)]
    \item If $n=1$, then $\abs{\witness_i(D)} = O(p)$ for $p \to \infty$ and
    \item if $n>1$, then $\abs{\witness_i(D)} = O(n^p)$ for $p \to \infty$.
  \end{enumerate}
  Furthermore, for any $p,n \in \N$ as above there is a clause set $N_{p,n}$ and a one-sided $X$-eliminating derivation $D_{p,n}$ from $N_{p,n}$ with $p$ purified pointed clauses where $n$ is the maximum number of $X$-literals such that
  \begin{enumerate}[(i)]
    \item if $n=1$, then $\abs{\alpha_i(D_{p,n})} = \Omega(p)$ for $p \to \infty$ and
    \item if $n>1$, then $\abs{\alpha_i(D_{p,n})} = \Omega(n^p)$ for $p \to \infty$.
  \end{enumerate}
\end{proposition}
\begin{proof}
  Note that an application of $T_S$ to witnesses $\seq{\witness}$ only increases if $S$ is a purified clause deletion step.
  Now let $P = \underline{L(\seq{t})} \lor C_P$ be a purified $X_i$-pointed clause in $D$.
  Since $D$ is one-sided we have that $P$ is one-sided, i.e., $P$ contains $X_i$ only with a single polarity.
  In turn $\resclosurepremise{P}{\seq{c}} = \mset{L(\seq{t})^\perp, \seq{c} \noeq \seq{t} \lor C_P}$.
  In $L(\seq{t})^\perp$ we have that $X_i$ occurs exactly once.
  In $C_P$ the total number of $\seq{X}$-literals is at most $n-1$.
  In total the number of $\seq{X}$-literals in $\resclosurepredicate{P}$ is at most $n$.
  Further let $q$ be the maximum size of $\resclosurepredicate{P}$ over all purified pointed clauses $P$ in $D$.
  Thus for predicate expressions $\seq{\beta} = (\beta_1, \dots, \beta_d)$ with $b := \max_{1 \leq i \leq d} \abs{\beta_i}$ we get 
  \begin{equation*}
    \label{eq.resclosurepredicate-inequality}
    \abs{\resclosurepredicate{P}[\seq{X} \leftarrow \seq{\beta}]} \leq n b + q.
  \end{equation*}
  
  Denote by $T_S(\seq{\beta})_i$ the $i$-the component of the tuple $T_S(\seq{\beta})$. Since purified clause deletion steps are the only derivation steps that increase the size of witnesses across $T_S$ we get
  \begin{equation}
    \label{eq.transformation-inequality}
    \abs{T_{\mathrm{PurDel}_P}(\seq{\beta})_i} \leq q + n b
    \text{ and } \abs{T_{S}(\seq{\beta})_i} \leq b \text{ for $S \neq \mathrm{PurDel}_P$}
  \end{equation}
  Remember that $\seq{\mathrm{wit}_m(D)} = (\lambda \seq{u_1}. W_1(\seq{u_1}), \dots, \lambda \seq{u_d}. W_d(\seq{u_d}))$ for some predicate variables $W_i$ and let $w := \max_{1 \leq i \leq d} \abs{W_i(\seq{u_i})}$.
  Recall, $\seq{\mathrm{wit}_{i-1}(D)} = T_{D_i}(\seq{\mathrm{wit}_i(D)})$ for all $1 \leq i \leq m$.
  By inductively applying \eqref{eq.transformation-inequality} for all derivation steps in $D$
  we get $\abs{\witness_i} \leq w_p$
  for a sequence $(w_k)_{k \in \N}$ which satisfies the recurrence relation 
  $$w_0 = w, \quad w_{k+1} = q + n w_k.$$
  The solution to this linear recurrence relation can be computed by standard methods as 
  $$w_k = w n^k + q\sum_{i=0}^{k-1} n^{i}.$$
  For $n=1$ we get $w_k = w + qk = O(k)$ and thus $\abs{\witness_i} = O(p)$ for $p \to \infty$.
  For $n>1$ we get $w_k = O(n^k)$ for $k \to \infty$ and thus $\abs{\witness_i} = O(n^p)$ for $p \to \infty$.

  We now show the second statement.
  Let $p,n \in \N$ as above.
  For $1 \leq i \leq p$ and $1 \leq j \leq n$ let $c_{i,j}$ be pairwise distinct constant symbols.
  We define the clause set 
  $$
  N_{p,n} := \mset{\Lor_{j=1}^n X({c_{i,j}}) \suchthat 1 \leq i \leq p}.
  $$
  For $1 \leq i \leq p$ consider the pointed clause $P_i := \underline{X({c_{i,1}})} \lor \Lor_{j=2}^n X({c_{i,j}})$.
  Note that all $P_i$ are one-sided.
  Also, since all $X$-literals in $N_{p,n}$ are positive, there are no resolvents that can be performed and thus every $P_i$ is purified in any subset of $N_{p,n}$.
  Thus, we can apply purified clause deletion to all $P_i$ in $N_{p,n}$ to get a one-sided $X$-eliminating $\mathcal{C}$-derivation $D_{p,n}$ from $N_{p,n}$ by
  \begin{prooftree}
    \Axiom$\mset{P_1, \dots, P_p}\fCenter$
    \RightLabel{$\mathrm{PurDel}_{P_1}$}
    \UnaryInf$\mset{P_2,\dots,P_p}\fCenter$
    \RightLabel{$\mathrm{PurDel}_{P_2}$}
    \UnaryInf$\vdots\fCenter$
    \RightLabel{$\mathrm{PurDel}_{P_{p-1}}$}
    \UnaryInf$\mset{P_p}\fCenter$
    \RightLabel{$\mathrm{PurDel}_{P_p}$}
    \UnaryInf$\emptyset\fCenter$
  \end{prooftree}
  Observe that $p$ is the number of purified pointed clauses in $D$ and that $n$ is the maximum number of $X$-literals of a purified clause in $D$.
  Note that 
  $$\resclosurepremise{P_i}{{d}} = \mset{\neg X({d}), {d} \noeq {c_{i,1}} \lor \Lor_{j=2}^n X({c_{i,j}})}$$ and thus $\resclosurepredicate{P_i} = \lambda {u}. X({u}) \lor (u \oeq {c_{i,1}} \land \Land_{j=2}^n \neg X({c_{i,j}}))$.
  Using this, we get for $1 \leq i \leq p$  
  $$\mathrm{wit}_{p-i-1}(D_{p,n}) = \lambda {u}. \mathrm{wit}_{p-i}[{u}] \lor ({u} \oeq {c_{p-i, 1}} \land \Land_{j=2}^n \neg \mathrm{wit}_{p-i}[{c_{p-i,j}}]).$$
  Let $w_i := \abs{\mathrm{wit}_i(D_{p,n})}$.
  Then $w_p = \abs{\lambda {u}. W({u})} = 2$ and 
  \begin{align*}
    w_{p-i-1} &= \abs{\lambda {u}. \mathrm{wit}_{p-i}[{u}] \lor \left({u} \oeq {c_{p-i, 1}} \land \Land_{j=2}^n \neg \mathrm{wit}_{p-i}[{c_{p-i,j}}]\right)} \\
    &= w_{p-i} + 3 + \sum_{j=2}^n (1 + w_{p-i})\\
    &= n w_{p-i} + n + 2.
  \end{align*}
  Solving this linear recurrence relation yields
  $$w_{p-i} = 2n^i + (n+2)\sum_{j=0}^{i-1} n^{i-j-1}.$$
  Now we get $\abs{\mathrm{wit}(D_{p,n})} = w_0 = 2n^p + (n+2)\sum_{j=0}^{p-1} n^{p-j-1}$.
  For $n=1$ we have $\abs{\mathrm{wit}(D_{p,n})} = 2 + 3p = \Omega(p)$ for $p \to \infty$.
  For $n > 1$ we have $\abs{\mathrm{wit}(D_{p,n})} = \Omega(n^p)$ for $p \to \infty$.
\end{proof}

\section{Improvement of Ackermann's Lemma}
\label{sec.connections-to-ackermanns-lemma}

Recall Ackermann's Lemma:
\begin{theorem}
  \label{ackermanns-lemma}
  Let $\varphi$, $\psi$ be first-order formulas where $X$ only occurs positively in~$\varphi$ and $X$ does not occur in $\psi$.
  Then
  \begin{align*}
    \Exists{X} (\varphi \land \Forall{\seq{u}} (X(\seq{u}) \limp \psi(\seq{u}, \seq{v}))) \iff \varphi[X \leftarrow \lambda \seq{u}. \psi(\seq{u}, \seq{v})]
  \end{align*}

  Let $\varphi$, $\psi$ be first-order formulas where $X$ only occurs negatively in~$\varphi$ and $X$ does not occur in $\psi$.
  Then
  \begin{align*}
  \Exists{X} (\varphi \land \Forall{\seq{u}} (\psi(\seq{u}, \seq{v}) \limp X(\seq{u}))) \iff \varphi[X \leftarrow \lambda \seq{u}. \psi(\seq{u}, \seq{v})]
  \end{align*}
\end{theorem}
\begin{proof}
  See, e.g.,~\cite[Lemma 6.1]{Gabbay08Second}.
\end{proof}

Note that in both cases Ackermann's Lemma provides a witness $\lambda \seq{u}. \psi(\seq{u}, \seq{v})$.
Therefore Ackermann's Lemma is a method for solving WSOQE.

We show that our method extends Ackermann's Lemma on clause sets:
\begin{proposition}
  \label{improvement-of-ackermanns-lemma}
  Let $N'$ be a finite clause set where $X$ only occurs positively and let $C(\seq{u})$ be a clause where $X$ does not occur.
  Further set $N := N' \union \mset{\neg X(\seq{u}) \lor C(\seq{u})}$.
  Then there is a one-sided $X$-eliminating derivation $D$ from $N$ such that $\mathrm{wit}(D) \iff \lambda \seq{u}. C(\seq{u})$, which is the witness produced by Ackermann's Lemma applied to $\Exists{X} N$.

  Let $N'$ be a finite clause set where $X$ only occurs negatively and let $C(\seq{u})$ be a clause where $X$ does not occur.
  Further set $N := N' \union \mset{X(\seq{u}) \lor C(\seq{u})}$.
  Then there is a one-sided $X$-eliminating derivation $D$ from $N$ such that $\mathrm{wit}(D) \iff \lambda \seq{u}. \neg C(\seq{u})$, which is the witness produced by Ackermann's Lemma applied to $\Exists{X} N$.
\end{proposition}
\begin{proof}
  Let $N' = \mset{C_1, \dots, C_n}$.
  Informally, the derivation proceeds as follows:
  First, form the resolution closure $\resclosure{P}{N'}$ with respect to the $X$-pointed clause $P = \underline{\neg X(\seq{u})} \lor C(\seq{u})$, resulting in a clause set $N_{\mathrm{Res}_P}$.
  If this terminates, delete $P$ by purified clause deletion, resulting in a clause set $N'_{\mathrm{Res}_P}$.
  The clause set $N'_{\mathrm{Res}_P}$ will contain $X$ only positively, so extended purity deletion is applicable, resulting in a clause set $N^\ast$.
  More formally, the $X$-eliminating derivation is given by $D=$
  \begin{prooftree}
    \Axiom$N\fCenter$
    \RightLabel{$\mathrm{Res}^\ast$}
    \UnaryInf$N_{\mathrm{Res}_P}\fCenter$
    \RightLabel{$\mathrm{PurDel}_P$}
    \UnaryInf$N'_{\mathrm{Res}_P}\fCenter$
    \RightLabel{$\mathrm{ExtPurDel}_X^{+}$}
    \UnaryInf$N^\ast\fCenter$
  \end{prooftree}
  Note that $\resclosurepremise{P}{\seq{c}} = \mset{X(\seq{c}), \seq{u} \noeq \seq{c} \lor C(\seq{u})}$ and thus 
  $$\resclosurepredicate{P} = \lambda \seq{v}. X(\seq{v}) \land \Forall{\seq{u}} (\seq{v} \noeq \seq{u} \lor C(\seq{u})) \iff \lambda \seq{v}. X(\seq{v}) \land C(\seq{v}).$$ 
  Now we have $\mathrm{wit}(D) = \resclosurepredicate{P}[X \leftarrow \lambda \seq{u}. \top] = \lambda \seq{v}. \top \land C(\seq{v}) \iff \lambda \seq{u}. C(\seq{u})$.

  To finish the proof it remains to show that forming the resolution closure terminates, i.e., to show that $\resclosure{P}{N'}$ is finite and to show that the clause set $N'_{\mathrm{Res}_P}$ contains $X$ only positively.
  One can show by induction on $\resclosure{P}{N'}$ that every clause $C \in \resclosure{P}{N'}$ only contains $X$ positively and whenever 
  \begin{prooftree}
    \AxiomC{$C$}
    \AxiomC{$P$}
    \BinaryInfC{$R$}
  \end{prooftree}
  is a constraint resolution inference between $C$ and $P$ upon $\neg X(\seq{u})$, then $R$ has one fewer positive occurrences of $X$ than $C$.
  This implies that the number of positive $X$-literals is a termination measure on $\resclosure{P}{N'}$, showing its finiteness.
  
  The other version can be proved by analogous arguments.
\end{proof}

\section{Examples}
\label{sec.examples}

\subsection{$X$-eliminating derivation with infinite witness}
\begin{example}
  \label{derivation-with-infinite-witness}
  Recall the clause set $N$ from \Cref{ex.main-example}
  \begin{align*}
    (1)\ B(a,v) \qquad (2)\ X(a) \qquad (3)\ B(u,v) \lor \neg X(u) \lor X(v) \qquad  (4)\ \neg X(c)
  \end{align*}
  and some resolvents from $N$ are
  \begin{align*}
    (5)\ a \noeq u \lor B(u,v) \lor X(v) \quad \text{($2$ with $3$)} \quad \text{and} \quad
    (6)\ & a \noeq c & \text{($2$ with $4$)}
  \end{align*}
  A derivation from $N$ different to $D$ from \Cref{ex.main-example} is given by $D' =$
  \begin{prooftree}
    \Axiom$\mset{1,2,3,4}\fCenter$
    \RightLabel{$\mathrm{PurDel}_{3.2}$}
    \UnaryInf$\mset{1,2,4}\fCenter$
    \RightLabel{$\mathrm{Res}_{2.1,4.1}$}
    \UnaryInf$\mset{1,2,4,6}\fCenter$
    \RightLabel{$\mathrm{PurDel}_{4.1}$}
    \UnaryInf$\mset{1,2,6}\fCenter$
    \RightLabel{$\mathrm{ExtPurDel}_X^{+}$}
    \UnaryInf$\mset{1,6}\fCenter$
  \end{prooftree}

  The $\mathrm{PurDel}_{3.2}$ derivation step is applicable from $\mset{1,2,3,4}$ since the only resolvent of $3.2$ with a clause in $N$ is the clause~$5$ which is redundant in $\mset{1,2,4}$.
  Since $\mset{1,6}$ does not contain $X$ we have that $D'$ is $X$-eliminating.
  However, $D'$ is not one-sided as the purified clause deletion step $\mathrm{PurDel}_{3.2}$ deletes the clause~$3$ which contains $X$ positively and negatively.
  The witness construction steps are:
  \renewcommand{\arraystretch}{1.2}
  $$
  \begin{array}{r@{\quad}l@{\quad}l@{\quad}l@{\quad}l}
    i\ & N(D')_i & D_{i+1}' & T_{D_{i+1}'}(\witness) &  \mathrm{wit}_{i}(D') \\
    \hline
    4\ & \mset{1,6} & - & - & W_X \\
    % \hline
    3\ & \mset{1,2,6} & \mathrm{ExtPurDel}_X^{-} & \lambda u. \bot & \lambda u. \bot  \\
    % \hline
    2\ & \mset{1,2,4,6} & \mathrm{PurDel}_{2.1} &  \resclosurepredicate{2.1}[X \leftarrow \witness] & \lambda u. u \oeq a \\
    % \hline
    1\ & \mset{1,2,4} & \mathrm{Res}_{2.1,4.1} & \witness & \lambda u. u \oeq a \\
    0\ & \mset{1,2,3,4} & \mathrm{PurDel}_{3.2} & \resclosurepredicate{3.2} & \resclosurepredicate{3.2}[X \leftarrow \lambda u. u \oeq a]
  \end{array}
  $$
  \renewcommand{\arraystretch}{1.0}%
  Note that $\resclosurepremise{3.2}{d}$ is infinite by \Cref{ex.resolution-closure} and therefore also $\resclosurepredicate{3.2}$ is an infinite expression.
  This means the witness $\mathrm{wit}(D')$ is infinite and has the form
  \begin{align*}
    \lambda u. &u \noeq c \\
    & \land \Forall{u'} (B(u,u') \lor u' \noeq c) \\
    & \land \Forall{u'}\Forall{u''} (B(u,u') \lor B(u',u'') \lor u'' \noeq c) \\
    & \land \dots
  \end{align*}
  This shows that from the same clause set there can be different $\seq{X}$-eliminating derivations where one leads to a finite witness and the other leads to an infinite witness.
\end{example}

\subsection{Multiple witnesses}

\begin{example}
  \label{ex.multiple-witnesses}
  Recall the example $\Exists{X} (X(a) \land \Forall{u} (X(u) \limp B(u)))$ from the introduction.
  Its clause normal form is
  $$
  (1)\ X(a) \quad (2)\ \neg X(u) \lor B(u).
  $$
  The only resolvent is $(3)\ B(a)$ (after constraint elimination).
  We can look at $4$ different derivations:
  \begin{align*}
    D^{(i)}&=\Axiom$\mset{1,2}\fCenter$
    \RightLabel{$\mathrm{Res}$}
    \UnaryInf$\mset{1,2,3}\fCenter$
    \RightLabel{$\mathrm{PurDel}_{1.1}$}
    \UnaryInf$\mset{2,3}\fCenter$
    \RightLabel{$\mathrm{PurDel}_{2.1}$}
    \UnaryInf$\mset{3}\fCenter$
    \DisplayProof
    &D^{(ii)}&=\Axiom$\mset{1,2}\fCenter$
    \RightLabel{$\mathrm{Res}$}
    \UnaryInf$\mset{1,2,3}\fCenter$
    \RightLabel{$\mathrm{PurDel}_{2.1}$}
    \UnaryInf$\mset{1,3}\fCenter$
    \RightLabel{$\mathrm{PurDel}_{1.1}$}
    \UnaryInf$\mset{3}\fCenter$
    \DisplayProof \\
    D^{(iii)}&=\Axiom$\mset{1,2}\fCenter$
    \RightLabel{$\mathrm{Res}$}
    \UnaryInf$\mset{1,2,3}\fCenter$
    \RightLabel{$\mathrm{PurDel}_{1.1}$}
    \UnaryInf$\mset{2,3}\fCenter$
    \RightLabel{$\mathrm{ExtPurDel}_{X}^{-}$}
    \UnaryInf$\mset{3}\fCenter$
    \DisplayProof
    &D^{(iv)}&=\Axiom$\mset{1,2}\fCenter$
    \RightLabel{$\mathrm{Res}$}
    \UnaryInf$\mset{1,2,3}\fCenter$
    \RightLabel{$\mathrm{PurDel}_{2.1}$}
    \UnaryInf$\mset{2,3}\fCenter$
    \RightLabel{$\mathrm{ExtPurDel}_{X}^{+}$}
    \UnaryInf$\mset{3}\fCenter$
    \DisplayProof
  \end{align*}

  They differ by the order in which purified clauses are deleted and whether extended purity deletion is used.
  We can compute $\resclosurepredicate{1.1} = \lambda u. X(u) \lor u \oeq a$ and $\resclosurepredicate{2.1} = \lambda u. X(u) \land B(u)$.
  The corresponding witnesses are
  \begin{align*}
    \mathrm{wit}(D^{(i)}) &= \resclosurepredicate{1.1}[X \leftarrow \resclosurepredicate{2.1}[X \leftarrow W]] = \lambda u. (W(u) \land B(u)) \lor u \oeq a \\
    \mathrm{wit}(D^{(ii)}) &= \resclosurepredicate{2.1}[X \leftarrow \resclosurepredicate{1.1}[X \leftarrow W]] = \lambda u. (W(u) \lor u \oeq a) \land B(u) \\
    \mathrm{wit}(D^{(iii)}) &= \resclosurepredicate{1.1}[X \leftarrow \lambda u. \bot] \iff \lambda u. u \oeq a \\
    \mathrm{wit}(D^{(iv)}) &= \resclosurepredicate{2.1}[X \leftarrow \lambda u. \top] \iff \lambda u. B(u)
  \end{align*}
  Note that all four witnesses are mutually non-equivalent.
  Furthermore, any instantiation of $W$ in $\mathrm{wit}(D^{(i)})$ and $\mathrm{wit}(D^{(ii)})$ gives a witness.
  For example we get further witnesses:
  \begin{align*}
    \mathrm{wit}(D^{(i)})[W \leftarrow \lambda u. \bot] &\iff \lambda u. u \oeq a \\
    \mathrm{wit}(D^{(i)})[W \leftarrow \lambda u. \top] &\iff \lambda u. B(u) \lor u \oeq a \\
    \mathrm{wit}(D^{(ii)})[W \leftarrow \lambda u. \bot] &\iff \lambda u. u \oeq a \land B(u) \\
    \mathrm{wit}(D^{(ii)})[W \leftarrow \lambda u. \top] &\iff \lambda u. B(u)
  \end{align*}
\end{example}

\subsection{Example with multiple variables}

We illustrate how the algorithm works with multiple variables on a simple example:
\begin{example}
  \label{ex.multiple-variables}
  Consider an example, adapted from~\cite{Wernhard17Boolean}, of a clause set $N$ with clauses
  $$
    (1)\ \neg X_1(u) \lor X_2(u) \quad (2)\ \neg A(u) \lor X_2(u) \quad (3)\ \neg X_2(u) \lor B(u)
  $$
  and one of its resolvents
  $$
  (4)\ \neg A(u) \lor B(u) \quad\text{(2 with 3)}
  $$
  One derivation from $N$ is given by $D=$
  $$
    \Axiom$\mset{1,2,3}\fCenter$
    \RightLabel{$\mathrm{PurDel}_{1.1}$}
    \UnaryInf$\mset{2,3}\fCenter$
    \RightLabel{$\mathrm{Res}$}
    \UnaryInf$\mset{2,3,4}\fCenter$
    \RightLabel{$\mathrm{PurDel}_{2.2}$}
    \UnaryInf$\mset{3,4}\fCenter$
    \RightLabel{$\mathrm{ExtPurDel}_{X_2}^{-}$}
    \UnaryInf$\mset{4}\fCenter$
    \DisplayProof
  $$

  Denote $(\alpha_i, \beta_i) := \seq{\mathrm{wit}_i(D)}$.
  Then the witness construction steps are (a dash `$-$' indicates that no change happened to the previous step)
  \renewcommand{\arraystretch}{1.2}
  $$
  \begin{array}{r@{\ }l@{\ }l@{\ }l@{\ }l}
    i & D_{i+1}& T_{D_{i+1}}(\alpha, \beta) &  \alpha_i & \beta_i \\
    \hline
    4 & \text{n/a} & \text{n/a} & W_1 & W_2 \\
    3  & \mathrm{ExtPurDel}_{X_2}^{-} & (\alpha, \lambda u. \bot) & - & \lambda u. \bot  \\
    2 & \mathrm{PurDel}_{2.2} &  (\alpha, \resclosurepredicate{2.2}[X_1 \leftarrow \alpha, X_2 \leftarrow \beta]) & - & \lambda u. A(u) \\
    1 & \mathrm{Res} & (\alpha, \beta) & - & \lambda u. A(u) \\
    0 & \mathrm{PurDel}_{1.1} & (\resclosurepredicate{1.1}[X_1 \leftarrow \alpha, X_2 \leftarrow \beta], \beta) & \lambda u. W_1(u) \land A(u)&  -
  \end{array}
  $$
  Note that $\resclosurepredicate{2.2} = \lambda u. X_2(u) \lor A(u)$ and $\resclosurepredicate{1.1} = \lambda u. X_1(u) \land X_2(u)$.
  Thus $\seq{\mathrm{wit}(D)} = (\lambda u. W_1(u) \land A(u), \lambda u. A(u))$.
  \renewcommand{\arraystretch}{1.0}%
\end{example}

\subsection{Example of Witness Produced by Implementation}

\begin{example}
  For the input clause set $N=$ 
  $$
  \mset{
  \neg B(a, v) \lor \neg B(v, w),
  \neg X(a),
  \neg B(u,v) \lor \neg X(u) \lor X(v),
  \neg X(c)
  }
  $$
  our implementation finds the following witness, if given a derivation limit of $20$:

  \begin{align*}
  \lambda u_0. 
    (&\ \forall v
      \forall v_0
      \forall v_1
      \forall v_2
        (\neg B(v_1, v_0) \lor  \neg B(v, c) \lor  \neg B(u_0, v_2) \lor  \neg B(v_2, v_1) \\
        &\qquad\qquad\qquad\quad\lor  \neg B(v_0, v)) \\
      &\quad \land 
      \forall v
      \forall v_0
      (\neg B(v, c) \lor  \neg B(v_0, v) \lor  \neg B(u_0, v_0)) \\
      &\quad \land 
      \forall v
      (\neg B(v, c) \lor  \neg B(u_0, v)) \\
      &\quad \land 
      B(a, u_0) \\
      &\quad \land 
      \neg B(u_0, c) \\
      &\quad \land 
      \forall v
      \forall v_0
      \forall v_3
      \forall v_1
      \forall v_2
      (\neg B(v_1, v_0) \lor 
        \neg B(v, c) \lor 
        \neg B(v_2, v_1) \lor 
        \neg B(u_0, v_3) \\
        &\qquad\qquad\qquad\qquad\qquad
        \lor \neg B(v_0, v) \lor 
        \neg B(v_3, v_2))\\
      &\quad \land 
      \forall v
      \forall v_0
      \forall v_1
      (\neg B(v, c) \lor  \neg B(v_0, v) \lor  \neg B(v_1, v_0) \lor  \neg B(u_0, v_1)) \\
      &\quad \land 
      c \neq u_0) \\
    &\lor a = u_0
  \end{align*}
\end{example}

\end{document}